\documentclass[11pt]{article}
\usepackage{url}
\usepackage[ruled,vlined,procnumbered]{algorithm2e}
\usepackage{ifpdf}
\usepackage[pdfstartview=FitH,pdfpagemode=None,backref,colorlinks=true,citecolor=blue,linkcolor=blue,urlcolor=black]{hyperref}\usepackage{breakurl} 
\usepackage{amsfonts}
\usepackage{latexsym,amssymb,amsmath,amsthm,mathrsfs}
\usepackage[ruled,vlined]{algorithm2e}
\usepackage{algorithmic}
\usepackage{verbatim,alltt} 
\usepackage[usenames,dvipsnames,table]{xcolor}
\usepackage{graphicx}
\usepackage{paralist}
\usepackage[margin=1.0in]{geometry}
\usepackage{pifont}

  \newtheorem{theorem}{Theorem}[section]
  \newtheorem{thm}[theorem]{Theorem}
  \newtheorem{lemma}[theorem]{Lemma}
  \newtheorem{proposition}[theorem]{Proposition}

  \newtheorem{corollary}[theorem]{Corollary}

  \newtheorem{claim}[theorem]{Claim}
  \newtheorem{obs}[theorem]{Observation}
  
  \newtheorem{definition}[theorem]{Definition}

  \newcommand{\eps}{\epsilon}

 \newcommand{\Z}{{\mathbb Z}}
 \newcommand{\N}{{\mathbb N}}

  \newcommand{\calG}{{\cal G}}
 
  \newcommand{\A}{{\cal A}}

 \newcommand{\polylog}{\operatorname{polylog}}
 
 \newcommand{\poly}{\operatorname{poly}}

\newcommand{\parent}{\operatorname{parent}}

\newcommand{\mcH}{\mathcal{H}}

  \newcommand{\view}{\textsc{View}}

\colorlet{d}{RedOrange!70!BurntOrange}
\colorlet{p}{PineGreen}
\colorlet{t}{MidnightBlue}
\colorlet{em}{Fuchsia}
\colorlet{tm}{BrickRed}

\newcommand{\mcA}{\mathcal{A}}

\newcommand{\ignore}[1]{}
\newcommand{\ignog}{\ignore}

\newcommand{\lowerapp}{\frac12n^{1-2\eps}}
\newcommand{\lowerprobes}{\eps n^{\eps}}

 \newcommand{\T}{T}
 \newcommand{\fuse}{\operatorname{fusion}}
 \newcommand{\expect}{\mathbb{E}}
\newcommand{\girth}{\operatorname{girth}(G)}

\newcommand{\namedref}[2]{\hyperref[#2]{#1~\ref*{#2}}}
\newcommand{\sectionref}[1]{\namedref{Section}{#1}}
\newcommand{\appendixref}[1]{\namedref{Appendix}{#1}}

\newcommand{\theoremref}[1]{\namedref{Theorem}{#1}}

\newcommand{\lemmaref}[1]{\namedref{Lemma}{#1}}
\newcommand{\claimref}[1]{\namedref{Claim}{#1}}
\newcommand{\corollaryref}[1]{\namedref{Corollary}{#1}}
\newcommand{\propref}[1]{\namedref{Proposition}{#1}}

\newcommand{\figref}[1]{\namedref{Figure}{#1}}

\newcommand{\equalityref}[1]{\hyperref[#1]{Equality~\eqref{#1}}}
\newcommand{\inequalityref}[1]{\hyperref[#1]{Inequality~\eqref{#1}}}
\newcommand{\algorithmref}[1]{\hyperref[#1]{Algorithm~\ref{#1}}}
\newcommand{\subsecref}[1]{\namedref{Subsection}{#1}}

\usepackage{color}
\usepackage{soul}

\renewcommand{\paragraph}[1]{\par\noindent\textbf{#1}}
\begin{document}
\def\thepage{}

	\title{
	\textbf{On the Probe Complexity of Local Computation Algorithms}
		}

	\author{Uriel Feige%
		\thanks{Weizmann Institute of Science, Rehovot, Israel. E-mail: {\tt  uriel.feige@weizmann.ac.il}. Supported in part by the Israel Science Foundation (grant No. 1388/16). Work partly done in Microsoft Research, Herzeliya, Israel. }
	\and Boaz Patt-Shamir%
	\thanks{Tel Aviv University, Tel Aviv, Israel. E-mail: {\tt  boaz@tau.ac.il}.
		Supported in part by the Israel Science Foundation (grant No. 1444/14).}
		\and Shai Vardi%
		\thanks{California Institute of Technology, Pasadena, CA, USA. E-mail: {\tt  svardi@caltech.edu}. Part of the research was carried out when Shai was a postdoctoral researcher at the Weizmann Institute of Science. Supported in part by the I-CORE in Algorithms Postdoctoral Fellowship.} }
	
	\maketitle
\begin{abstract}
	
The Local Computation Algorithms (LCA) model  is a computational
model aimed at problem instances with huge inputs and output. For
graph problems, the input graph is accessed using probes: \emph{strong probes (SP)} specify a vertex $v$ and receive as a reply a list of $v$'s neighbors, and \emph{weak probes (WP)}  specify a vertex $v$ and a port number $i$ and receive as a  reply  $v$'s $i^{th}$ neighbor.  Given a local query (e.g., ``is a certain vertex in
the vertex cover of the input graph?''), an LCA should compute the
corresponding local output (e.g., ``yes'' or ``no'') while making
only a small number of probes, with the requirement that all local
outputs form a single global solution (e.g., a legal vertex cover).
We study the probe complexity of LCAs
  that are required to work on graphs that may have arbitrarily large
  degrees. In particular, such LCAs are  expected to  probe the graph a number of times that is significantly smaller than the maximum, average, or even minimum  degree.

  For weak probes, we focus on the weak coloring problem~\cite{NN90}. Among our results we show a separation between weak 3-coloring and weak 2-coloring for deterministic LCAs: $\log^* n + O(1)$ weak probes suffice for weak 3-coloring, but $\Omega\left(\frac{\log n}{\log\log n}\right)$ weak probes are required for weak 2-coloring.

   For strong probes, we consider randomized LCAs for vertex cover and maximal/maximum matching. Our negative results include showing that there are graphs for which finding a \emph{maximal} matching requires $\Omega(\sqrt{n})$ strong probes. On the positive side, we design a randomized LCA that finds a $(1-\eps)$ approximation to \emph{maximum} matching in regular graphs, and uses $\frac{1}{\eps}^{O\left( \frac{1}{\eps^2}\right)}$ probes, independently of the number of vertices and of their degrees.

\end{abstract}

\newpage
\pagenumbering{arabic}


\section{Introduction}

In classical algorithmic models, an algorithm is given an input
and is required to compute an output. When dealing with truly massive data, such as the Internet, just reading the entire input may turn out to be impossible or
impractical. The model of \emph{local computation algorithms} (LCAs), as studied
by Rubinfeld et al.~\cite{RTVX11}, proposes the following idea.
An algorithm in the LCA model is required to produce only a specified part of the output, and
is expected to access only a ``small'' part of the input (without any pre-processing).
For example, an LCA for maximal independent set (MIS) is given a vertex ID as a \emph{query},
and is expected to return a ``yes/no'' answer, indicating whether or not the queried vertex is in the MIS; all replies to queries must be consistent with the same MIS.

The focus of this paper is LCAs for graph problems; the input is a graph, which the LCA is allowed to probe. For the purpose of understanding our notions of probes, it helps to think of the graph as being represented as a two dimensional $n$ by $d+1$ array (where $d$ is the maximum degree). Rows are labeled from $1$ to $n$ by the vertex names. In a given row $v$, the cell $(v,0)$ specifies the degree $d_v$ of $v$, the cell $(v,j)$ for $1 \le j \le d_v$ specifies the ID of the neighbor connected to $v$'s $j$th port (the order of these neighbors is arbitrary, not necessarily from smallest to largest), and cells $(v,j)$ for $d_v  < j \le d$ contain $0$. A {\bf strong probe (SP)}  specifies a row number $v$ and gets the entire row in response. A {\bf weak probe (WP)} specifies a single cell $(v,j)$ and gets its content in response. Weak probes are identical to the graph probing characterization of Goldreich and Ron~\cite{GR02}. Previous work typically considered graphs of constant bounded degree, and did not differentiate between 
strong and weak probes; as we consider graphs with unbounded degrees, it is natural to distinguish between them.

\ignog{The LCA is allowed
to \emph{probe} the input graph.
	We distinguish between two types of probes:\footnote{Previous work on LCAs does not typically discriminate between these two types of probes; if the maximal degree, $d$ of the input graphs considered is bounded, and we are interested in probe complexity polynomial or even exponential in the degree, the two types give similar probe complexities. }
	\begin{itemize}
		\item {\bf Strong probes (SP)} By specifying a vertex ID, the algorithm obtains a list of the IDs of that vertex's neighbors.
		\item {\bf Weak probes (WP)} 
		By specifying a vertex ID $v$, the algorithm learns the number of neighbors $v$ has. By specifying a pair $(v,i)$, the algorithm obtains the ID of $v$'s $i^{th}$ neighbor (we assume there is some fixed order on the neighbors of each vertex).
	\end{itemize}
}
LCAs are useful in scenarios when the input is so large that we may not even be able to compute the entire solution, yet there is a global function that we wish to query small parts of. Well known examples include \emph{locally-decodable codes} (LDCs) (e.g., \cite{KT00, Yekh12}) and local reconstruction (e.g., \cite{JR11, SS10}). LDCs are   error-correcting codes that allow a single bit of the original message to be decoded with high probability by querying a small number of bits of a  codeword. Local reconstruction involves recovering the value of a function $f$ for a particular input $x$ given oracle access to a closely related function $g$. LCAs have recently been applied  to solving convex problems 
in a distributed fashion~\cite{PCVW17}. Traditional methods for solving distributed optimization, such as iterative descent methods (e.g.,~\cite{Low2002}) and consensus methods (e.g.,~\cite{Blondel05}) require global communication, and any edge failure or lag in the system immediately affects the entire solution, by delaying computation or causing it not to be computed at all; furthermore, if the network changes in a small way, the entire solution needs to be recomputed. If an LCA is used, most of the system remains unaffected by local changes or failures. Hence LCAs can be used to  make systems more robust to edge failures, lag, and dynamic changes.   There are other situations when LCAs may be useful - say we wish to perform some computation on each of the vertices of an MIS of some huge graph. LCAs allow us to be able to begin work on some vertices before the entire MIS is computed, and guarantee that the local replies to the queries will be consistent with the same global solution, that will be available at some point in the future. See \cite{VardiPhD} for  more on the applications of LCAs.


One of the main challenges of designing LCAs is bounding their  probe complexity - the worst-case number of probes the LCA needs to make to the graph in order to reply to a single query. Although there are other complexity measures for LCAs (see the related work section), in this work we primarily focus on probe complexity, and only briefly remark upon other measures.

Noting that in many real-world applications, graphs have vertices of high degree (see, e.g., \cite{KCF11}), 
it is natural to ask whether one can design LCAs with low probe complexity
for graphs in which the maximal, average, or even minimal degree is large.
Another natural question is more ``local'' in nature: can we design an LCA that, for vertex $v$ with degree $d_v$, uses $o(d_v)$ probes?
There are known LCAs for  MIS, maximal matching, approximate maximum matching, coloring and other problems, e.g., \cite{ARVX11,EMR14,LRY15,MRVX12,MV13,RV16}: the probe complexity 
in these results is typically  \emph{exponential} in the maximal degree $d$ (\cite{LRY15}
gives an LCA for approximate maximum matching
  with probe complexity polynomial in $d$). 
None of the above results appears to be useful for graphs of high degree;
for example, if we would like to have a $\polylog(n)$-probe LCA, the above results only hold when the maximal (or average) degree is at most $O{(\log\log{n})}$. A notable exception is the algorithm of Levi et al.~\cite{LRY15}, whose complexity is roughly $\tilde{O}{(d^{1/\eps^2}\log^2{n})}$, and therefore gives a $\polylog(n)$-probe LCA for $d = O(\polylog{n})$.

\medskip

 In this work, we consider LCAs for graph problems  when the degrees are large.  We focus on the \emph{number of probes} that a vertex needs to make in order to compute its value in the solution: Can we even hope for an LCA whose probe complexity is polylogarithmic in $n$ if the minimal degree is $n^{\Omega(1)}$?
From the algorithmic perspective,  the main challenge is that in high-degree graphs,
we cannot afford to obtain information about all the neighbors of a vertex: using strong probes, we can only know their IDs; using weak probes - not even that. In many cases, this is sufficient reason
to preclude (even SP) LCAs with sub-degree probe complexity altogether. Consider, for example, a problem that is considered easy in most natural computational models: given a tree, mark all nodes that have a leaf as a neighbor. In the closely related distributed LOCAL model (e.g., \cite{Peleg00}), each vertex can send an unbounded message to all of its neighbors and  do unlimited computation in each round. The complexity of algorithms in the LOCAL model is measured by the number of rounds required for every vertex to compute its own value in the solution. For this problem (the leaf-neighbor-marking problem), a single round suffices, but it is easy to see that an LCA requires $\Omega(d_v)$ (weak or strong) probes per query (where $d_v$ is the degree of the queried vertex).

 \subsection{Our Contributions}

 Since we do not want to probe all neighbors, we  need to carefully choose a small set of neighbors that we \emph{do} probe.  In~\sectionref{sec:sample}, we discuss several natural methods of sampling a single neighbor,  a ``parent''.
 In subsequent sections, these parent sampling techniques are used in LCAs that we design for the problem of \emph{weak-coloring} (a coloring in which every non-isolated vertex $v$ has at least one neighbor colored differently than $v$).

 In~\sectionref{sec:color}, we give tight upper and lower bounds for weak $3$-coloring. Our algorithm is deterministic and uses weak probes, whereas our lower bound holds also for randomized LCAs that may use strong probes.

 \begin{theorem}\label{thm:weak3upper}
	There exists a deterministic WP LCA for  weak $3$-coloring that uses $\log^*{n}+O(1)$  probes.
\end{theorem}

\begin{theorem}\label{thm:weak3lower}
	Every (randomized or deterministic) SP LCA for weak $3$-coloring of the cycle graph requires $\Omega(\log^*{n})$  probes.
\end{theorem}

 In \sectionref{subsec:det2} we show how to augment the algorithm  of \sectionref{sec:color} in order to reduce the number of colors to $2$:

\begin{theorem}\label{thm:detupper2}
	There exists a deterministic  WP LCA  for  weak $2$-coloring that uses $\log^*{n}+2d_v+O(1)$ weak probes, where $d_v$ is the degree of the queried vertex.
\end{theorem}

We show that some dependence on the degree is unavoidable, at least when using weak probes.

\begin{theorem} \label{thm:detlower2} Any WP LCA for weak $2$-coloring $d$-regular graphs with $d =O\left( \frac{\log{n}}{\log\log{n}}\right)$ requires at least $d/2$ probes.
\end{theorem}


In~\sectionref{subsec:rand2}, we design a randomized LCA for weak $2$-coloring, whose probe complexity is independent of the maximal degree and show how it can be implemented in both the strong and weak probe models.
\begin{theorem}\label{thm:genweakrand}
	There exists a randomized WP LCA for weak $2$-coloring that uses $\Theta\left( \frac{\log^2 {n}}{\log\log{n}}\right) $ probes, and a randomized SP LCA for weak $2$-coloring that uses $\Theta\left( \frac{\log{n}}{\log\log{n}}\right) $ probes.
\end{theorem}

In~\sectionref{sec:VC}, we give a lower bound for vertex cover in the strong probe model. Specifically, we show that for high degree graphs, 
many strong probes are necessary to approximate a minimal vertex cover to any interesting precision.
\begin{theorem}
	\label{thm:lower}
	For any $\eps<\frac{1}{2}$, any randomized  SP LCA that computes a vertex cover  whose size is a $(\lowerapp)$-approximation to the size of the  minimal vertex cover requires at least $\lowerprobes$ probes.
\end{theorem}

A corollary of  \theoremref{thm:lower} is the following.

\begin{corollary}
	Any SP LCA for maximal matching on arbitrary graphs requires $\Omega({n}^{1/2-o(1)})$ probes.
\end{corollary}

 In~\sectionref{sec:match1}, we describe an LCA that finds a matching that is a $(1-\epsilon)$-approximation to the  maximum matching, for regular (and approximately-regular, see~\sectionref{sec:match1} for a definition) graphs, using a constant number of probes.

\begin{theorem}\label{thm:dregularbest}
	There exists an SP LCA that finds a $(1-\eps)$-approximate maximum matching in expectation on $d$-regular graphs that uses $\left( \frac{1}{\eps}\right)^{O\left( \frac{1}{\eps^2}\right)}$ probes per query.
\end{theorem}

In~\sectionref{sec:match}, we show that for graphs with sufficiently high girth and degree, polynomially (in $1/\eps$) many strong probes suffice.
\begin{theorem}
	\label{thm:girth}
	There exists an SP LCA that finds a  $(1-\eps)$-approximate maximum matching in expectation on $d$-regular graphs of girth $g$, with $d \ge \frac{1}{\epsilon}$ and $g \ge \frac{1}{\epsilon^3}$, that uses ${O(1/{\eps^7})}$ probes per query.
\end{theorem}



\subsection{Overview of our Techniques}

In~\sectionref{sec:sample} we discuss deterministic and randomized  methods of sampling a ``parent'' for each vertex. No matter how the parents are chosen, the resulting graph is a disjoint set of directed subgraphs, each one containing exactly one cycle. The main technical content of~\sectionref{sec:sample} is the analysis of the diameter of these subgraphs, depending on the choice of parent selecting scheme.

	 In~\sectionref{sec:color} we address the problem of weak $3$-coloring. Our LCA (the proof of \theoremref{thm:weak3upper}) is based on the following approach. Given an arbitrary graph, the directed subgraphs formed by the parent relation (as above) span all vertices, and each of their connected components has at least two vertices. Hence any weak 3-coloring of every component separately induces a weak 3-coloring of the whole graph. Each component has one cycle, but it turns out that the $3$-coloring algorithm for rooted trees of Goldberg, Plotkin and Shannon \cite{GPS88} can be adapted in order to legally 3-color (and hence also weakly 3-color) such a component. When implemented as an LCA, the upper bound of $\log^* n + O(1)$ on the number of probes follows from a similar upper bound on the number of rounds of the (modified) algorithm of \cite{GPS88}.
	 To prove a nearly matching lower bound (\theoremref{thm:weak3lower}) we  use a reduction to the lower bound of Linial \cite{Linial92} for distributed algorithms that legally 6-color a cycle. Adapting lower bounds from the distributed setting to the LCA setting also involves an argument of G{\"{o}}{\"{o}}s et al.~\cite{GHLMS15}.

In \sectionref{subsec:det2} we show how to augment the algorithm  of \sectionref{sec:color} in order to reduce the number of colors to $2$, thus proving \theoremref{thm:detupper2}. It takes only one more probe to transform the weak 3-coloring to a weak $2$-coloring that is legal for all vertices except for those vertices that do not serve as parents (which we refer to as {\em leaves}). The final step involves changing the colors of (some) leaves. In order to determine whether a vertex $v$ is a leaf, we probe all of its neighbors, making the probe complexity linear in the degree of the queried vertex. 

One natural method of proving lower bounds for LCAs
is by reduction to the distributed LOCAL model, as was done in \cite{GHLMS15} (and as we do in the proof of \theoremref{thm:weak3lower}). The relationship between LCAs and distributed algorithms has been studied before (e.g., \cite{EMR14,PR07,RV16,RTVX11}) -- given a distributed algorithm to a problem that takes $t$ rounds, one immediately obtains an LCA that uses $d^{\,t}$ probes (where $d$ is the maximal degree), by probing all nodes within radius $t$. The inverse
reduction doesn't work, as an LCA may probe remote (disconnected) nodes. Consider,
for example, the following artificial problem: each vertex has to color itself blue if the node with ID~$1$ has an odd number of neighbors, and red otherwise. An LCA for this problem needs a single probe, while a distributed algorithm requires time proportional to the graph's diameter.
G{\"{o}}{\"{o}}s et al.~\cite{GHLMS15} show that for many natural problems, probing undiscovered vertices does not help. 
But even if we consider only local probes, the best lower bound we can hope for using such
a reduction is the distributed \emph{time} lower bound: a lower bound of $t$ rounds  in the distributed model implies a  lower bound of $t$ probes in the LCA model (recall that an upper bound of $t$ rounds in the distributed model implies an upper bound of $d^t$ probes!). This suggests that we may need new tools to obtain stronger lower bounds. In the proof of Theorem~\ref{thm:detlower2}, we iteratively construct families of $d$-regular graphs, where graphs in family $F_i$ (for $i \le d/2$) contain roughly $d^i$ vertices, and we show that weak 2-coloring of graphs in family $F_i$ requires at least $i$ deterministic weak probes. Every graph $G$ in family $F_i$ is composed of $d$ disjoint copies of graphs $H_1, \ldots, H_d$ from family $F_{i-1}$ and two auxiliary vertices $a_i$ and $b_i$. From each graph $H_j$ a single edge is removed, and instead one of its endpoints is connected to $a_i$ and the other to $b_i$. Intuitively, it is reasonable to expect that $a_i$ cannot decide on a color before it knows the color of at least one of its neighbors. Likewise, it is reasonable to expect that each neighbor, being a member of a graph in $F_{i-1}$, requires (by an induction hypothesis) $i-1$ probes. The combination of these two non-rigorous claims would imply that $a_i$ requires at least $i$ probes (one to determine the name of one of its neighbors, and $i-1$ probes so as to determine the color of that neighbor). Turning this informal intuition into a rigorous proof is nontrivial, and this is the main content of our proof of Theorem~\ref{thm:detlower2}.

		In~\sectionref{subsec:rand2}, we design a randomized LCA for weak $2$-coloring. A key aspect in our randomized LCA is that each vertex first chooses a random temporary color. This induces a weak 2-coloring for most vertices of the graph: each vertex whose temporary color differs from the temporary color of a neighbor can keep its color. Extending this weak 2-coloring to the remaining vertices is done by associating a parent with each vertex, with the intended goal that the color of the vertex will differ from the color of the parent (determining the color of the parent uses an inductive process). This aspect has several different implementations, leading to different probe complexities. Namely, for an arbitrary parent choice, the number of strong probes is $\Theta(\log{n})$. If we implement the arbitrary choice using weak probes, the probe complexity is $\Theta(\log^2{n})$. For a more clever randomized choice of parent, we get that the strong probe complexity is $\Theta\left(\frac{\log{n}}{\log\log{n}} \right)$, and the weak probe complexity is $\Theta\left(\frac{\log^2{n}}{\log\log{n}} \right)$. All these bounds on probe complexity hold with high probability.
		
	We note that our results do not prove a separation between the complexities of deterministic and randomized WP LCAs for weak $2$-coloring (although we conjecture that there is one), as our lower bound that is linear in the degree is proved only for regular graphs of degree at most $O\left(\frac{\log{n}}{\log\log{n}} \right)$, and our upper bound for randomized WP LCAs for weak $2$-coloring is $O\left(\frac{\log^2{n}}{\log\log{n}} \right)$.

	Randomized LCAs generally use a pseudo-random generator in order to limit the number of bits that they use, while ensuring consistency  (e.g., \cite{ARVX11,LRY15,RV16}). In order to explore the theoretical limitations of the probe complexity of LCAs, we assume that our randomized LCAs have unbounded access to random bits. Nevertheless,  we show in~\sectionref{subsec:rand2} that one can implement the randomized WP LCA for weak 2-coloring (the one using the arbitrary parent choice scheme) using a pseudo-random generator with a seed of length $O(\log{n})$.

In the proof of \theoremref{thm:lower} (\sectionref{sec:VC}) we construct a family of bipartite graphs in which a large subset of vertices have ``almost'' the same view at distance $2$.  Exactly one of these vertices, $v_0$, needs to be added to the vertex cover; however, there are many vertices for which a small number of strong probes does not suffice in order to verify that  they are not $v_0$, and hence they must add themselves to the vertex cover. In our construction we use vertex naming schemes based on low degree polynomials in order to ensure that certain vertices do not share many neighbors.

In~\sectionref{sec:match1} we give an LCA for approximate maximum matching in regular graphs. We first describe an LCA for $(1-\eps)$ matching on graphs of degree bounded by $d$, that uses at most $\left(d + \frac{1}{\eps}\right)^{O\left( \frac{1}{\eps^2}\right)}$ probes per query. (Alternatively, if we wish to have a better dependency on $\epsilon$ and are willing to have a dependency on $n$, then an LCA of Even, Medina and Ron~\cite{EMR14} has  probe complexity  $d^{O\left( \frac{1}{\eps}\right) } + O\left( \frac{1}{\eps^2}\right) \log^*{n}$.) Our LCA $\A$ is a simple variation on a randomized LCA of Yoshida, Yamamoto and Ito~\cite{YYI12}: whereas~\cite{YYI12} do not limit the number of probes used by their LCA but instead analyse and provide upper bounds on the expected number of probes used by their LCA (expectation taken both other choice of random edge and randomness of the LCA), we run essentially the same LCA, but with a strict upper bound on the number of probes. This upper bound is a factor  of $\frac{d}{\eps}$ larger than the expectation. Markov's inequality implies that this gives a $(1-\eps)$-approximation to the maximum matching in expectation.

To obtain an LCA for a $d$-regular graph $G$, if $d < \frac{1}{\eps^2}$  we use LCA $\A$. If $d > \frac{1}{\eps^2}$, we sparsify the graph: for some universal constant $c$ (independent of $\epsilon$), each edge remains in the graph with probability $\frac{c}{d\eps}$, and then all vertices that still have degree higher than $\frac{2c}{\eps}$ are removed from the graph. This results in a graph $G'$ of degree bounded by $\frac{2c}{\eps}$. Every matching in $G'$ is a matching in $G$. Moreover, we show that the expected size of a maximum matching in $G'$ (expectation taken over choice of random sparsification) is at least $(1 - \frac{\eps}{2})$ times the size of the maximum matching in $G$.  Hence it suffices to find a $(1-\frac{\eps}{2})$-approximate maximum matching in the bounded degree graph $G'$, and it will serve as a $(1 - \eps)$-approximate maximum matching in $G$.  A sufficiently large matching in $G'$ can be found by $\A$, using  $\left(\frac{1}{\eps}\right)^{O\left( \frac{1}{\eps^2}\right)}$ probes to $G'$ per query. To implement $\A$ on $G'$, but using only probes to $G$, we show that each strong probe to $G'$ can be simulated by $O(1/\eps)$ strong probes to $G$.

In~\sectionref{sec:match}, we show that we can find a $(1-\eps)$-approximation to the maximum matching for regular graphs with sufficiently high degree and girth using polynomially (in $1/\eps$) many probes. Gamarnik and Goldberg~\cite{GG10} show that the randomized  greedy algorithm finds a $(1-\eps)$ approximation to the maximum matching on regular graphs with sufficiently high degree and girth. Similarly to~\sectionref{sec:match1}, if the vertex degrees are sufficiently small (say, below $\frac{1}{\eps^3}$), we can use an LCA of~\cite{YYI12} (an implementation of the randomized greedy algorithm that uses in expectation $O(d)$ probes in graphs of maximum degree $d$), while placing a strict upper bound on the number of probes (this upper bound is a factor of $\frac{1}{\eps^{O(1)}}$ larger than the expectation). 
If the degrees are large, our approach is once again to sparsify the graph $G$ prior to using the LCA of~\cite{YYI12} (modified to have a strict upper bound on the number of probes). Unfortunately, the resulting graph $G'$ is only nearly regular but not actually regular. This requires us to extend the result of~\cite{GG10} from regular graphs to nearly regular graphs. We do so without need to refer to the proofs of~\cite{GG10}, by the following approach.
We add imaginary edges to $G'$, making it regular (while maintaining high girth). The LCA now runs on a regular graph, and hence the bounds of~\cite{GG10} apply. The new problem that arises is that the matching that is output by the LCA might contain imaginary edges. However, we prove that the expected fraction of imaginary edges in that matching is similar to their fraction within the input graph (our proof uses both the high girth assumption and the fact that the randomized greedy algorithm is local in nature). This, combined with the fact that the fraction of imaginary edges in the input graph was small (because $G'$ is nearly regular), implies that the imaginary edges can be discarded from the solution without significantly affecting the approximation ratio.

\medskip

\subsection{Related Work}
In this work, we focus on the probe complexity of LCAs; however, there are several criteria by which one can measure the performance of LCAs. Rubinfeld et al.~\cite{RTVX11} focus on the time complexity of LCAs: how long it takes to reply to a query; Alon et al.~\cite{ARVX11} emphasize the space complexity, in particular, the length of the random seed used (randomized LCAs need a global
random seed to ensure consistency). Expanding upon this, Mansour et al.~\cite{MRVX12} describe LCAs for maximal matching and other problems that use polylogarithmic (in the size of the graph; exponential in the degree) time and space when the degrees of the graph are bounded by a constant. %
Even, Medina and Ron~\cite{EMR14}, focusing on probe complexity,
give deterministic LCAs for MIS, maximal matching and $(d+1)$-coloring for graphs of degree bounded by a constant $d$, which use  $d^{O(d^2)}\log^*{n}$ probes.  Fraigniauld, Heinrich and Kosowski~\cite{FHK16} investigate local conflict coloring, a general distributed labeling problem, and use their results to improve the probe complexity of $(d+1)$-vertex coloring to approximately $d^{O(\sqrt{d})} \log^* n$ probes. Mansour, Patt-Shamir and Vardi~\cite{MPV15} introduce a unified model, that takes into account all four complexity measures: \emph{probe complexity}, \emph{time complexity} (per query) and space complexity, divided into \emph{enduring memory} (in all known LCAs, this includes only the random seed) and \emph{transient space} (the computational space used per query). They
show that it is possible to obtain LCAs for which all of these are independent of $n$ for certain problems, such as a $(1-\epsilon)$ approximation to a maximal acyclic subgraph, using $d^{O(1/\eps)}$ probes. Few papers address super-constant degrees: Levi, Rubinfeld and Yodpinyanee~\cite{LRY15} give LCAs for MIS and maximal matching for graphs of  degree $2^{O(\sqrt{\log\log n})}$, using an improvement of
Ghaffari~\cite{Ghaf16}.
Reingold and Vardi~\cite{RV16} give LCAs for MIS, maximal matching and other problems that apply to graphs that are sampled from some distribution. This limitation allows them to address graphs with higher maximal degree, as long as the average degree is $O(\log\log{n})$, and the tail of the distribution is sufficiently light. If we restrict ourselves to LCAs that use polylogarithmic time and space, the approximate maximum matching LCA of  \cite{LRY15} accommodates graphs of polylogarithmic degree. Levi, Ron and Rubinfeld~\cite{LRR14} describe
an LCA that constructs spanners using a number of probes polynomial in $d$.


There are few explicit impossibility results LCAs.  G{\"{o}}{\"{o}}s et al.~\cite{GHLMS15} show that any LCA for MIS requires $\Omega(\log^*{n})$ probes, by showing that probing vertices that have not yet been discovered is not useful. This implies that, on a ring, the number of probes that an LCA needs to make is ``roughly the same'' as the number of rounds required by a distributed LOCAL algorithm,
implying that the lower bounds of Linial~\cite{Linial92} holds for LCAs as well. Levi, Ron and Rubinfeld~\cite{LRR14} show that an LCA that determines whether an edge belongs to a sparse spanning subgraph requires $\Omega(\sqrt{n})$ probes. Feige, Mansour and Schapire~\cite{FMS15}, adapting a lower bound from the property testing literature~\cite{GR02}, show  that approximating the minimum vertex cover in bounded degree bipartite graphs within a ratio of $1 + \eps$ (for some explicit fixed $\eps > 0$) cannot be done with $o(\sqrt{n})$ probes.

 The distributed lower bound for vertex cover  by Kuhn, Moscibroda and Wattenhofer \cite{KMW04,KMW16} shows that no $k$-round distributed algorithm for Vertex Cover has an approximation ratio better than $\Omega\left( \frac{n^{c/k^2}}{k}\right) $ or $\Omega\left( \frac{d^{c'/k}}{k}\right) $ for  positive constants $c$ and $c'$. The main idea behind the lower bound of~\cite{KMW04,KMW16} is that two vertices, one of which should clearly be in the vertex cover, and the other of which should not, have the same view at any distance up to $k$. Compare this with our lower bound construction for the proof of~\theoremref{thm:lower}: by requiring that two vertices only have \emph{almost} the same view, we are able to  leverage the difference between distributed algorithms and LCAs to obtain a stronger lower bound.

Weak coloring was introduced by Naor and Stockmeyer \cite{NS95}. They show that there exists a LOCAL distributed algorithm that requires $\log^* d +O(1)$ rounds for weak $2$-coloring graphs of maximal degree $d$, assuming the degree of every vertex is odd. 
In contrast, they show that there is no constant time LOCAL algorithm for weak $c$-coloring all graphs with vertices with even degrees, for constant $c$. Our deterministic weak 2-coloring LCA (\theoremref{thm:detupper2}) uses $ \log^* n+O(d_v)$ probes, but when cast within the LOCAL model it takes $\log^* n+O(1)$ rounds (independently of $d_v$).

LCAs can be used as subroutines in approximation algorithms e.g., \cite{CRT05,PR07,NO08,YYI12}. The goal of such algorithms is to output an approximation to the size of the solution to some combinatorial problem (such as Vertex Cover, Maximum Matching, Minimum Spanning Tree), in time sublinear in the input size. In particular, if one has an LCA whose running time is $t$ that solves some problem, one can obtain an approximation to the size of the solution (with some constant probability) by executing this LCA on a sufficiently large (but constant) number of vertices $k$ chosen uniformly at random, to obtain an approximation algorithm whose running time is $kt$ (see e.g., \cite{NO08,VardiPhD} for more details).

Sampling techniques similar to the one we use in~\sectionref{sec:match1} for sparsifying regular graphs are used by Goel, Kapralov and Khanna~\cite{GKK10} -- each edge of a regular bipartite graph is kept independently with some probability, and the remaining graph still has a perfect matching w.h.p. The same authors also use an adaptive sampling technique to further reduce the running time of finding a perfect matching in a regular bipartite graph to $O(n\log{n})$~\cite{GKK13}.

The concept of LCAs is related to but should not be confused with local algorithms as in~\cite{Andersen10,AGPT16,ST13}. The difference is that these local algorithms do not require the output for different probes to satisfy a global consistency property, but rather to satisfy some local criteria. For example, the goal might be for each vertex to output a small dense subgraph that contains it, without requiring two different vertices to agree on whether they share the same dense subgraph or not.



\section{Preliminaries}
\label{sec:prelims}

We denote the set of integers $\{1,2, \ldots, n\}$ by $[n]$. All logarithms are base $2$.
Our input is a simple undirected graph $G=(V,E)$,  in which every vertex has an ID and all IDs are distinct, $|V|=n$. The neighborhood of a vertex $v$, denoted $N(v)$, is the set of vertices that share an edge with $v$: $N(v)=\{u:(u,v)\in E\}$. The \emph{degree} of a vertex $v$,  is $|N(v)|$. The \emph{girth} of $G$, denoted $\girth$ is the length of the shortest cycle in $G$.

\paragraph{LCAs.} Our definition of LCAs is slightly different from previous definitions in that it focuses on probe complexity. We do this so as not to introduce unnecessary notation. See \cite{MPV15, VardiPhD} for definitions that also take into account other complexity measures.

\begin{definition}[Probe]\label{def:probe} We assume that the input graph is represented as a two dimensional $n$ by $d+1$ array, where $d$ is the maximum degree. Rows are labeled from $1$ to $n$ by the vertex names. For any $v$, the cell $(v,0)$ specifies the degree $d_v$ of $v$, the cell $(v,j)$ for $1 \le j \le d_v$ specifies the name of the neighbor connected to $v$'s $j$th port. Cells $(v,j)$ for $d_v  < j \le d$ contain $0$. We define strong and weak probes as follows.
	\begin{itemize}
		\item A \emph{strong probe} (SP) specifies the ID of a vertex $v$; the reply to the probe is the entire row corresponding to $v$ (namely, the list of all neighbors of $v$).
		\item A \emph{weak probe} (WP) specifies a single cell $(v,j)$ and  receives its content (namely, only the $j^{th}$ neighbor of $v$).
		\end{itemize}
\end{definition}

We note that knowing that $u$ is $v$'s $i^{th}$ neighbor does not give us information regarding which one of $u$'s ports $v$ is connected to. This property is crucial for the proof of~\theoremref{thm:detlower2}.

\begin{definition}[Local computation algorithm]\label{def:lca}
	A deterministic {\em $p(n)$-probe 
		local computation algorithm} $\mcA$ for a computational problem is an
	algorithm that receives an input of size $n$. Given a query $x$,
 $\mcA$ makes at most $p(n)$ probes to the input in order to reply.
 $\mcA$ must be \emph{consistent}; that is, the algorithm's replies to all of the
	possible queries combine to a single feasible solution to the problem.

	A randomized {\em $(p(n),s(n), {\delta(n)})$-local computation algorithm} $\mcA$ differs from a deterministic one in the following aspects. Before receiving its input, it is allowed to write $s(n)$ random bits (referred to as the random {\em seed}) to memory.\footnote{In this work, we generally allow the random seed to be unbounded. In Section~\ref{subsec:rand2} we remark upon the seed length/probe tradeoff that we can achieve for the algorithm described therein.} Thereafter, it must behave like a deterministic LCA, except that when answering queries it may also 
	read and use the random seed. For any input $G$, $|G|=n$, the probability (over the choice of random seed)  that there exists a query in $G$ for which  $\mcA$ uses more than $p(n)$ probes is at most $\delta(n)$,
	which is called $\mcA$'s failure probability.
\end{definition}

When LCA $\mcA$ is given input graph $G=(V,E)$ and queried on vertex $v \in V$, we denote this by $\mcA(G,v)$.

An LCA $\mcA$ is said to \emph{require} $k$ probes on a graph $G=(V,E)$, if there is at least one query  for which $\mcA$ uses $k$ probes. We say that an LCA $\mcA$ requires $k$ probes for a family $F$ of graphs, if for some graph $G \in F$, $\mcA$ requires $k$ probes. 
\medskip

 \paragraph{Graph Problems.}
 We examine several well-studied graph problems:
 \begin{compactitem}
 	\item
 	{\em Weak $c$-Coloring:} Given a graph $G=(V,E)$ and an integer $c$, find an assignment $col:V \rightarrow [c]$, such that for each non-isolated vertex $v \in V$ there exists a vertex  $u \in N(v)$ such that $col(v) \neq col(u)$. Weak $2$-coloring can be viewed also as a \emph{domatic partition} -- partitioning the vertices of the graph into (at least) two sets such that each one of them is a dominating set. Weak 3-coloring is a natural weakening of weak 2-coloring, but no longer gives a domatic partition.
 	
 	  To avoid confusion, we refer to the problem of assigning each vertex a color from $[c]$ such that no two neighbors have the same color (usually simply called $c$-coloring) by \emph{proper} $c$-coloring. 
 \item\emph{(Minimum) Vertex Cover (VC):} Given a graph $G=(V,E)$, find a  set of vertices $S \subseteq V$ of minimum cardinality, such that for every edge $e=(u,v) \in E$, at least one of $\{u,v\}$ is in $S$.
\item   {\em Maximal Matching:} Given a graph $G=(V,E)$, find an edge set $S \in E$ such that no two edges in $S$ share a vertex, and there is no edge $e \in E \setminus S$ such that $e$ has no neighbors in $S$.
\item
 {\em Maximum Matching:} Given a graph $G=(V,E)$, find an edge set $S \in E$ of maximum cardinality such that no two edges in $S$ share a vertex.

 \end{compactitem} 
   \begin{definition}[Approximation algorithm]
  	Given a maximization  problem over graphs and a real number $0 \leq \alpha \leq 1$, a (possibly randomized) $\alpha$-approximation algorithm $\mcA$ is guaranteed, for any input graph $G$, to output a feasible solution whose expected value is at least an $\alpha$ fraction of the value of an optimal solution (in expectation over the random bits used by $\mcA$).\footnote{The definition of approximation algorithms to minimization problems is analogous, with $\alpha \geq 1$.}
  \end{definition}

\begin{definition}[Unicyclic tree, unicyclic forest]
	 A  set $S$ of vertices and $|S|$ directed edges such that each vertex has one outgoing edge is called a \emph{unicyclic tree}.  A disjoint set of unicyclic trees  is called a \emph{unicyclic forest}. 
\end{definition}
Observe that in a unicyclic tree, there is a single cycle, and there are no edges from a vertex on the cycle to a vertex not on the cycle. Furthermore, unless self-loops are allowed, there are no isolated vertices in a unicyclic  forest.

\section{Local Exploration}
\label{sec:sample}

We analyze classes of directed graphs with out-degree one that result from procedures by which local algorithms explore the neighborhood of a given vertex in a graph $G$. We assume that vertices have IDs, which might be adversarial. For randomized algorithms, we allow the algorithm to replace the original IDs by random IDs, chosen from a large enough range so that collisions are unlikely. (If collisions do occur, they can be resolved by appending the original ID to the random ID, with negligible effect on our bounds, because collisions are so unlikely.) Let $N(v)$ denote the open neighborhood of $v$ and $N^+(v)$  the closed neighborhood (including $v$ itself).

We consider several ways by which a vertex $v$ chooses an outgoing edge to a parent vertex $p(v)$. The three deterministic schemes are:

\begin{enumerate}

\item {\bf Arbitrary}. A neighbor $u \in N(v)$ is chosen arbitrarily, and becomes $p(v)$.

\item {\bf Lowest neighbor}. The vertex $u \in N(v)$ of smallest ID ($ID(u) = \min_{w \in N(v)}[ID(w)]$) becomes $p(v)$.

\item  {\bf Lowest ID}. The vertex $u \in N^+(v)$ of smallest ID becomes $p(v)$. Here, if $v$ has a lower ID that any of its neighbors then $v$ is its own parent ($p(v) = v$), and the directed graph has a self-loop.

\end{enumerate}

For randomized schemes, we have:
\begin{enumerate}
	\item  {\bf Random}. A neighbor $u \in N(v)$ is chosen uniformly at random.
	\item  {\bf Random lowest neighbor}. Same as lowest neighbor, but with random IDs.
	
	\item  {\bf Random lowest ID}. Same as lowest ID, but with random IDs.
	\end{enumerate}

We note that it is possible to implement all schemes in the strong probe model using a single probe. In the weak probe model,  the arbitrary and random schemes can be implemented using a single probe, but the rest use a number of probes that can be as large as the degree of $v$.

\begin{proposition}\label{prop:unicyclic }
In any of the above schemes, the graph decomposes into unicyclic  trees.%
\end{proposition}

\begin{proof}
Every vertex has one outgoing edge. Hence the number of edges equals the number of vertices in each connected component.
\end{proof}

The following two lemmas are useful for understanding the lowest neighbor and lowest ID schemes. In the following, ``directed path'' refers to the path implied by the directed edges created by the \emph{parent} operation. 

\begin{lemma}
\label{lem:nonincrease1}
	In the lowest ID scheme, every directed path has non-increasing IDs (decreasing, except for self loops). Moreover, no neighbor (in $G$) of a path vertex $u$ has lower ID than any of the vertices reachable from $u$ via this directed path.

\end{lemma}

\begin{proof}
Every directed path is non-decreasing from the definition of the scheme, and the fact that the IDs are unique implies that, except for self-loops, the paths are decreasing.
Consider any vertex $v$ that is a neighbor of $u$, and any vertex $w$, reachable from $u$ via the directed path. 
Then it must hold that $w \leq  p(u) \leq v$, where the first inequality is due to the first part of the lemma, and the second due to the choice of $p(u)$.
\end{proof}

\begin{lemma}
	\label{lem:nonincrease2}
	 In the lowest neighbor scheme, the vertices at odd locations (or even locations) of any directed path have non-increasing IDs. Moreover, no neighbor (in $G$) of a path vertex $u$ has lower ID than any of the vertices reachable from $u$ via the directed path at an odd distance from $u$.
 \end{lemma}
\begin{proof}
	The proof of the first part of the lemma is the following: consider $3$ vertices on a path - $v_1 \rightarrow v_2 \rightarrow v_3$ (possibly $v_1=v_3$). As $p(v_2)=v_3$, $v_3\leq v_1$. 
The proof of the second part of the lemma is identical to the proof of the second part of Lemma~\ref{lem:nonincrease1}.
	\end{proof}

\begin{proposition}
For the arbitrary parent scheme, the cycle in each component may be of arbitrary length (but greater than one), for the lowest neighbor scheme the cycle is of length two, and for lowest ID scheme the cycle is of length one (self loop).
\end{proposition}

\begin{proof}
For the arbitrary and lowest neighbor schemes, as a vertex cannot choose itself as its parent, the cycle must be of length at least two. In the lowest neighbor scheme, the vertex with the smallest ID in the connected component, $u$ has an outgoing edge to some vertex $w$, who in turn must have an outgoing edge to $u$ (otherwise $u$ would not have the lowest ID in the  connected component).  In the lowest ID scheme, the cycle must be of length one, as it is a self loop at the vertex with the lowest ID in the connected component.
\end{proof}

We refer to the vertex of lowest ID on the unique cycle of a component as the {\em root} of the component. (When random IDs are used, the root can be the cycle vertex of either lowest original ID or lowest random ID -- both options work.)

We analyze how long it takes a vertex $v$ to discover the root of its own component by the following procedure: follow the directed path that starts at $v$ until a cycle is completed. We refer to this as the longest directed path from $v$.

\begin{proposition}
\label{pro:diameter2}
There are graphs of diameter~$2$ for which the longest directed path in the lowest ID scheme is of length $n-2$.
\end{proposition}

\begin{proof}
Let $G$ be a graph composed of a path with $n-1$ vertices with IDs from $n-1$ to~1. The longest directed path is then of length $n-2$. To make the diameter equal to~2, add one vertex with ID $n$ connected to all other vertices.
\end{proof}

\begin{proposition} There are graphs for which  the longest directed path in the random lowest ID scheme has length $\Omega(\log n)$ with constant probability, where the probability is taken over the choice of random IDs.
\end{proposition}

\begin{proof}

Let $n$ be such that $\sqrt{n}+1$ is a power of $2$. Consider a graph $G$ composed of $\sqrt{n}$ disjoint copies of the following graph $H$. The graph $H$ is composed of a sequence of layers, with $2^i$ vertices in layer $i$, for $0 \le i \leq  \log\left(\frac{\sqrt{n}}{2}+\frac{1}{2}\right) = L$; every two successive layers comprise a complete bipartite graph. In any such $H$, the lowest random ID vertex has probability (greater than) $1/2$ of being in layer $L$. Conditioned on this, the probability that the lowest random ID vertex in layers $\{0,\ldots, L-1\}$ is at level $L-1$ is (greater than) $1/2$, and so on. Hence, the probability that the directed path from the root has length $L$ is at least $\frac{1}{2^L} \approx \frac{2}{\sqrt{n}}$. By the union bound, the probability that the length of the longest directed path in $G$ is at most $L-1$ is at most $\left( 1 - \frac{2}{\sqrt{n}}\right) ^{\sqrt{n}} \approx e^{-2}.$
\end{proof}

\begin{thm}
For any graph $G$,  the longest directed path has length at most $O(\log n)$  with high probability (over the random choice of IDs) in the random lowest neighbor scheme. Moreover, if $G$ is of bounded degree $d$, then with high probability the longest directed path has length at most $O_d\left( \frac{\log n}{\log\log n}\right) $ ($O_d$ signifies that the constant in the $O$ notation depends on $d$).
\end{thm}

\begin{proof}
	Denote the random ID of vertex $v$ by $r(v)$. 
	Assume w.l.o.g. that the random ID are drawn uniformly at random from $[0,1]$. We show that with high probability, any directed path that has length $c \log{n}$ has a vertex with ID at most $\frac{1}{n^2}$. This suffices, as the probability that any vertex has ID less than $n^{-2}$ is at most $n^{-1}$, by the union bound.  Label the even vertices on a directed path that begins at $v_1$ (i.e., $v_1$ has no incoming edges) by $v_1,\ldots, v_{\ell}$, where $2\ell$ is the length of the path (we assume the path has even length for simplicity).
	Let $X_i$ be the random variable denoting the ID of $v_i$.
	If $p(v_i)$ has $j$ neighbors whose IDs are lower than $r(v_i)$, the expected value of $X_{i+1}$ is $v_i/(j+1)$, as the IDs are chosen uniformly at random, and the marginal distribution of the IDs is uniform over $[0, v_i]$; hence $$\expect[X_{i+1}|X_i=x] \leq x/2.$$ 
	Let $Y_i$ denote the event that $X_{i+1}> X_i/2$. Notice that these events are independent, conditioned on there being a path of length $c \log{n}$ (it is easy to see that they are independent if one thinks of computing the random IDs of vertices only when they are encountered). By the Chernoff bound,
	
	$$\Pr\left[ \sum_i Y_i<\frac{c}{4} \log{n}\right]  \leq e^{\frac{-c\log{n}}{16}}.$$

Hence, setting $c$ to be large enough, it holds that with probability at least $1-1/n^2$, there will be at least $2\log{n}$ indices $i$ such that $X_{i+1}<X_i/2$, hence $X_{\ell}$ is at most $n^{-2}$. Taking a union bound over all possible paths (there are at most $n$ such paths, one from each vertex), completes the proof of the first part.

For graphs of maximum degree $d$, there are at most $nd^{2k}$ paths of length $2k$. The probability of a decreasing sequence in the even locations is $1/k!$.
A simple union bound gives that the probability that there is at least one directed $p$-path of length exactly $k$ is at most (by Stirling's inequality)
$$\frac{nd^{2k}}{k!} \leq 
n \left( \frac{(de)^2}{k}\right) ^k.$$
As every path of length at least $k$ contains a path of length exactly $k$, this bound holds for paths of length at least $k$. 
Setting $k = O\left( \frac{\log{n}}{\log{\log{n}}}\right) $ completes the proof.
\end{proof}


\section{Weak $3$-Coloring}
\label{sec:color}



An LCA for weak $c$-coloring receives as a query a vertex $v$ and is required to return $v$'s color in a legal weak $c$-coloring. In this section,  we describe a deterministic WP LCA for weak $3$-coloring that uses $\log^*{n}+O(1)$ probes, independent of the maximal degree. In~\sectionref{subsec:det2} we use this LCA as a subroutine in a  deterministic  LCAs for weak $2$-coloring, that uses $\log^*{n} +2d_v+O(1)$ probes, where $d_v$ is the degree of the queried vertex.


Recall the  proper  $3$-coloring algorithm of Goldberg, Plotkin and Shannon \cite{GPS88} (which we denote by GPS) for rooted forests. The GPS algorithm has two parts; in the first, a proper $6$-coloring is found in $\log^*{n}+O(1)$ rounds, (using the symmetry breaking technique of Cole and Vishkin~\cite{CV86}), and in the second, the number of colors is reduced to $3$ in three more rounds. Let $\T(G)$ be the number of rounds by which the GPS algorithm is guaranteed to terminate for graph $G$.  Set $\T_n = \max\{\T(G):|G|=n\}$; that is, $\T_n$ is the maximal round complexity of the GPS algorithm on graphs of size $n$. By \cite{GPS88},  $\T_n=\log^*{n}+O(1)$.\footnote{In~\cite{GPS88} the bound is stated as $O(\log^*{n})$. It is easy to verify (see also~\cite{BE13}) that $\T_n=\log^*{n}+O(1)$.}

We would like to create a directed spanning forest that does not have isolated vertices, and simulate the GPS algorithm on it. This would guarantee a proper $3$-coloring of the forest, implying a weak $3$-coloring of the graph. Unfortunately, it is not clear how to build such a forest using $O(\log^*{n})$ weak probes.\footnote{With strong probes, we could use the smallest neighbor scheme (the smaller ID vertex of the $2$-cycle would be the root).} Instead, we do the following: Each vertex chooses an arbitrary parent, creating a spanning unicyclic  forest (see \propref{prop:unicyclic  }). We then simulate a modified version of the GPS algorithm on this graph, which we call MGPS. The modifications are minor - first,  as unicyclic  forests do not have roots, we do not concern ourselves with them. The second is simply a change in the stopping condition - the GPS algorithm stops when the number of colors is $6$; as this is a global property that cannot be verified locally, we  run the algorithm for $T_n$ rounds. One can verify that  the proofs of \cite{GPS88} hold for MGPS on unicyclic  forests; see~\appendixref{app:goldberg} for more details.

Our weak $3$-coloring LCA is the following (\algorithmref{alg:lca3col}). Given a graph $G$, we first probe the graph in order to discover  a (subgraph of a) spanning unicyclic  forest:
 For each vertex $v$, let $\parent(v)$ denote an arbitrary neighbor of $v$. 
 Our subgraph is $v$ and all of its ancestors up to distance  (at most) $\T(G)+1$. 
 The MGPS algorithm is then simulated for $v$ (we expand on how below), and the color of $v$ is returned.
 

\vspace{10pt}

\begin{algorithm}[h]
	\caption{\textsc{Weak $3$-coloring WP LCA}}\label{alg:lca3col}
	\SetKwInOut{Input}{Input}\SetKwInOut{Output}{Output}\SetKwInOut{Inquiry}{Inquiry}
	
	\Input{$G=(V,E)$, query $v$}
	\Output{a color}
	\BlankLine
	
	$G'=(V',E')$, $V'=\emptyset$, $E'=\emptyset$\;
	$dist=1$\;
	$u=v$\;
	$V' \leftarrow v$\;
	\While{$dist \leq \T_n+1$}
	{
			\tcc{The following can be replaced with 
				any arbitrary choice of parent}
		Let $\parent(u)$ be $u$'s neighbor at port $1$\;
			$E' \leftarrow (u,\parent(u))$\;
		\If{$\parent(u)  \in V'$}
		{
			Break\;}
		$V' \leftarrow \parent(u)$\;
	
		$u=\parent(u)$\;
		$dist=dist+1$\;
	}
		$i=1$\;
		\While{$i\leq \T_n$}{
		Execute round $i$ of the MGPS algorithm on all vertices at distance at most $\T_n-i$ from $v$ in $V'$\;
		$i=i+1$\;}
		Return $v$'s color.
	
\end{algorithm}

\begin{claim}\label{claim:color}
	Given a graph $G=(V,E)$ and query $v$, the color returned for a vertex $v$ by  \algorithmref{alg:lca3col} is the same as that of the MGPS algorithm.
	\end{claim}
\begin{proof}
	We need to show that at the end of the {\bf while} loop of~\algorithmref{alg:lca3col}, the color of $v$ is identical to the color it would have been assigned by running MGPS on $G$. To this end, we show that (1) the $\T_n+1$ probes give us sufficient information to perform the necessary computations and (2) it suffices to execute round $i$ of the MGPS algorithm on vertices at distance at most $\T_n+1-i$ from $v$. We prove (2) by ``reverse'' induction. For the base case, $i=\T_n$, clearly it suffices to simulate the MGPS algorithm for $\T_n$ rounds on vertex $v$. For the inductive step, assume that it suffices to simulate MGPS for $i$ rounds on vertices at distance at most $\T_n-i$ from $v$. In order to do this, we only need to know the color of the vertices at distance at most  $\T_n-i+1$ from $v$ at time $i-1$. 
	
	To prove (1), we note that in order to simulate MGPS on these $\T_n$ vertices, we need the initial color of the vertices as well as the parent of the furthest vertex from $v$ - a total of $\T_n+1$ vertices, hence we require $\T_n+1$ probes.
	\end{proof}

 {
	\renewcommand{\thetheorem}{\ref{thm:weak3upper}}
	
\begin{theorem}
	Given a graph $G=(V,E)$, and query $v \in V$,  \algorithmref{alg:lca3col} returns a color for $v$ that is consistent with some valid weak $3$-coloring using $\log^*{n}+O(1)$ weak probes.
\end{theorem}
	\addtocounter{theorem}{-1}
}

\begin{proof}
	
	From Claim~\ref{claim:color}, every vertex has at least one neighbor with a different color (its parent); hence the resulting coloring is a valid weak $3$-coloring.
The probe complexity is immediate, as we only probe the graph in the initialization phase, in which we perform $\T_n+1 = \log^*{n}+O(1)$ probes.
\end{proof}

This result is asymptotically optimal:
 {
	\renewcommand{\thetheorem}{\ref{thm:weak3lower}}

\begin{theorem}
	Every SP LCA for weak $3$-coloring a ring requires $\Omega(\log^*{n})$ probes.
\end{theorem}
	\addtocounter{theorem}{-1}
}

\begin{proof}
	From \cite{Linial92}, every distributed, possibly randomized, LOCAL algorithm for $c$-coloring a ring requires $\Omega(\log^*{n})$ rounds, for any constant $c$. Using the reduction of \cite{GHLMS15}, this implies that every LCA for $c$-coloring a ring requires $\Omega(\log^*{n})$ probes. 
	Suppose that there exists a weak $3$-coloring LCA that uses $p=o(\log^*{n})$ probes. Then we can convert it to a proper $6$-coloring LCA on a ring using $2p$ more probes,  as follows. Each vertex discovers the color of its neighbors, using one query (and hence at most $p$ probes) per neighbor. If two adjacent vertices  have the same color, the one with the lower ID adds $3$ to its color. 	 The new coloring is a proper $6$-coloring: if two adjacent vertices have the same color after this step, it must hold that  they had the same color before and either both or neither changed their color. In either case, it means that at least one of them, $w$, had another neighbor with the same color. But then $w$ had the same color as both of its neighbors, a contradiction.
	\end{proof}

\section{Deterministic WP Weak $2$-Coloring}
\label{subsec:det2}
In \subsecref{subsec:detupper}, We describe a WP LCA for weak $2$-coloring an arbitrary graph in $\log^*{n} +2d_v+O(1)$ probes, where $d_v$ is the degree of the queried vertex. In \subsecref{subsec:lower} we show that some dependence on the degree cannot be avoided. Combined with \theoremref{thm:weak3lower}, this implies that the LCA of \subsecref{subsec:detupper} is asymptotically optimal.

\subsection{Upper Bound}
\label{subsec:detupper}

We extend \algorithmref{alg:lca3col} to  a weak $2$-coloring LCA,~\algorithmref{alg:lca2cold}. The high level description of the algorithm is the following: we use~\algorithmref{alg:lca3col} to generate a spanning unicyclic  forest and color it with the colors $\{0,1,2\}$ (recall that although the output of~\algorithmref{alg:lca3col} is a weak coloring, on every unicyclic  tree it is in fact a proper  coloring). We then recolor all vertices that have the color $2$ with the complement of their parent's color; this ensures that all non-leaf vertices are colored legally (a \emph{leaf} is a vertex that is not the parent of any other vertex). Finally, we recolor all leaves with the complement of their parent's color. 

\vspace{10pt}
\begin{algorithm}[h]
	\caption{\textsc{Deterministic Weak $2$-coloring WP LCA}}\label{alg:lca2cold}
	\SetKwInOut{Input}{Input}\SetKwInOut{Output}{Output}\SetKwInOut{Inquiry}{Inquiry}
	
	\Input{$G=(V,E)$, query $v$}
	\Output{$0/1$}
	\BlankLine

		\tcc{Step $1$}
	Execute Algorithm~\ref{alg:lca3col} on $(G,v)$\;
			\tcc{The following can be replaced with 
		any arbitrary choice of parent}
	Let $\parent(u)$ be $u$'s neighbor at port $1$\;
	\tcc{Step $2$}
	\If{$color(v)=2$}
	{\If {$color(\parent(v))=1$}
		{$color(v)=0$\;}
		\Else{$color(v)=1$\;}
		}
	\tcc{Step $3$}
	Probe $v$'s neighbors to find out who their parents are\;
	\If {$v$ is a leaf (no neighbor of $v$ has $v$ as its parent)}
	{\If {$color(\parent(v))=1$}
		{$color(v)=0$\;}
		\Else{$color(v)=1$\;}
	}
\end{algorithm}

\begin{claim}\label{claim:step1}Step $2$ of \algorithmref{alg:lca2cold} guarantees that every non-leaf has a different color from either its parent or child.\end{claim}
\begin{proof}
	In Step $2$, every vertex $v$ that has the color $2$ recolors itself $1$ if the color of $\parent(v)$ is $0$, and $0$ otherwise. Hence $v$ has a different color than its parent, and both $v$ and its parent obey the weak coloring rule. We need to verify that $v$'s child, denoted $w$,  also has a different color than its parent or child (assuming $w$ is not a leaf). Assume w.l.o.g. that $v$ was recolored $0$. Vertex $w$ was not colored $2$ at the end of Step $1$, as the output of~\algorithmref{alg:lca3col} is a legal (not weak) $3$-coloring of the subgraph. If $w$ was colored $1$, we are done. Otherwise, $w$ was colored $0$, but its child was colored a different color. If its child was colored $2$, it is now colored $1$. 
\end{proof}
 Step $3$ guarantees that leaves have a different color from their parents,  and does not create any new color conflicts; hence  \algorithmref{alg:lca2cold} yields a legal weak $2$-coloring. Step $3$ uses $2d_v-2$ probes: $d_v-1$ to discover all of $v$'s neighbors (the parent had already been discovered), and $1$ more probe per neighbor, to check if $v$ is its parent. 
This gives
 {
	\renewcommand{\thetheorem}{\ref{thm:detupper2}}
\begin{theorem}
 \algorithmref{alg:lca2cold}  is a deterministic  weak $2$-coloring WP LCA that uses  $\log^*{n}+2d_v+O(1)$ weak probes, where $d_v$ is the degree of the queried vertex.
\end{theorem}
	\addtocounter{theorem}{-1}
}

From ~\theoremref{thm:weak3lower} we know that every LCA for deterministic weak $2$-coloring requires $\Omega(\log^*{n})$ probes. We now show that some dependence on the degree is unavoidable.

\subsection{Lower Bound}
\label{subsec:lower}
We prove the following.
 {
	\renewcommand{\thetheorem}{\ref{thm:detlower2}}
	\begin{theorem}
Let $d \ge 2$ be an arbitrary integer and let $q$ be the largest integer satisfying $2q < d + 3$. Then there is some $n \le 2(d+1)^{q}$ such that every weak 2-coloring WP LCA  that works on all $d$-regular graphs of size $n$ requires at least $q$ weak probes. In particular, there exists a constant $c>0$ such that every WP LCA for weak 2-coloring requires at least $\frac{d}{2}$  probes whenever $d \le \frac{c\log n}{\log\log n}$.
\end{theorem}
	\addtocounter{theorem}{-1}
}

Before proving the theorem 
 we present some useful principles.

\begin{proposition}
\label{pro:renaming}
For given positive integers $n,k,t$ and a family $F$ of graphs with $n$ vertices, suppose that every weak 2-coloring algorithm for $F$ requires at least $k$ probes. Then the same holds also if the name space is changed from $[1,n]$ to $[t+1, t+n]$.
\end{proposition}

\begin{proof}
Suppose for the sake of contradiction that on name space $[t+1, t+n]$ there is a weak 2-coloring algorithm that uses fewer than $k$ probes. Then run the same algorithm on name space $[1,n]$, adding $t$ to each name.
\end{proof}

\begin{lemma}
\label{lem:adversary}
For given positive integers $n,k$ and a family $F$ of graphs with $n$ vertices, suppose that every weak 2-coloring algorithm for $F$ requires at least $k$ probes. Then for every algorithm $W$ for weak 2-coloring of $F$ that uses at most $k$ probes, there is an adversary $A_{F,W}$ that does the following:

\begin{enumerate}

\item Selects a query vertex $v$, $1 \le v \le n$.

\item Adaptively answers any legal sequence of $k-1$ probes, where the answer to probe $j$ for $1 \le j \le k-1$ may depend on all probes up to and including probe $j$.

\item Exhibits two graphs from $F$ (or possibly, the same graph but with different naming of the vertices) that are consistent with all $k-1$ probe answers, and such that in one graph $W(v) = 0$ leads to an inconsistent weak coloring and in the other $W(v) = 1$  leads to an inconsistent weak coloring. (Namely, the information provided by the answers to the $k-1$ probes is insufficient in order to determine whether $W$ would color $v$ with color~0 or with color~1.) 

\end{enumerate}
\end{lemma}

\begin{proof}
Given $F$ and $W$ consider a two player game between a player $P$ who makes probes, and a player $A$ who answers the probes. The game starts with player $A$ announcing a vertex name $v$. It then continues for $k-1$ rounds, where in each round $P$ makes a legal probe and $A$ provides an answer. Player $A$ wins if at the end of the game there are two graphs from $F$ (or possibly, the same graph but with different naming of the vertices) that are consistent with all probe answers, and such that in one graph $W(v)$ must be $0$ and in the other $W(v)$ must be $1$. Being a finite full information sequential game, either $P$ or $A$ have a winning strategy. Suppose that $P$ has a winning strategy. Then this strategy gives rise to a weak 2-coloring algorithm for $F$ with $k-1$ probes (because at that point $W(v)$ can be the same for all namings for graphs of $F$ that are consistent with the probe answers), contradicting the premise of the lemma. Hence $A$ has a winning strategy, and this can serve as $A_{F,W}$, proving the lemma.
\end{proof}


We are now ready to prove~\theoremref{thm:detlower2}.

\begin{proof}
We shall construct by induction a sequence $F_1, \ldots, F_q$ of families of graphs, such that every graph $G_i \in F_i$ (for $1 \le i \le q$) is $d$-regular and has $n_i \le 2(d+1)^{i}$ vertices, and every algorithm for weak 2-coloring that works on the entire family $F_i$ requires at least $i$ weak probes.

{\bf The base case of the induction.} The family $F_1$ contains a single graph $G_1$ -- the complete bipartite graph on $2d+2$ vertices $K_{d+1,d+1}$, minus a perfect matching. As required by the induction invariant, it indeed is $d$-regular, has at most $n_1 \le 2(d+1)$ vertices, and weak $2$-coloring it requires at least one probe. (If there are zero probes, then for some color, say~0, there are at least $d+1$ named vertices that always have this color. If we name the vertices of $G_1$ such that these 0-colored vertices correspond to one vertex and its $d$ neighbors, the resulting coloring is not a legal weak 2-coloring.)

{\bf The inductive step.} Suppose that we have family $F_i$ that satisfies the inductive hypothesis, for some $1 \le i < q$. We now explain how to derive family $F_{i+1}$ from it. A graph in $F_{i+1}$ will be composed of $d$ graphs $H_1, \ldots, H_d$ from $F_i$ and two special vertices $a_{i+1}$ and $b_{i+1}$. In each of the graphs $H_j$ we remove one edge, and connect $a_{i+1}$ to one of its endpoints and $b_{i+1}$ to its other endpoint. Thus the graph obtained has $n_{i+1} = dn_i + 2 \le (d+1)n_i \le 2(d+1)^{i+1}$ vertices and is $d$ regular. Doing so in all possible ways (namely, every graph from $F_i$ can serve as each $H_j$, and every edge in $H_j$ can be removed) gives the set of graphs that constitute $F_{i+1}$.

Let $k$ be smallest integer such that there is a weak 2-coloring for $F_{i+1}$ that uses at most $k$ probes, and let $W_{i+1}$ be a weak 2-coloring algorithm for $F_{i+1}$ that uses at most $k$ probes. It remains to prove that $k \ge i+1$. For the sake of contradiction, suppose that $k < i+1$. 

By definition, $W_{i+1}$ specifies for each queried vertex $v \in [n_{i+1}]$ a probing algorithm.
Let $W_i$ be the algorithm for $F_i$ in which each  queried vertex in $[n_i]$ uses its algorithm from $W_{i+1}$.

\begin{lemma}
\label{lem:W_i}
Algorithm $W_i$ as defined above is a weak 2-coloring algorithm for $F_i$. 
\end{lemma}

\begin{proof}
Suppose for the sake of contradiction that there is some graph $G_i \in F_i$ which $W_i$ fails to weakly 2-color. Let $v \in [n_i]$ be a vertex colored by $W_i$ in the same color as its neighbors. We now show that there is an input $G_{i+1} \in F_{i+1}$ in which $v$ is colored by $W_{i+1}$ in the same color as its neighbors, contradicting the fact that $W_{i+1}$ be a weak 2-coloring algorithm for $F_{i+1}$.

Run $W_i$ on $d+1$ vertices from $G_i$, namely, on $v$ and on its neighbors. Altogether this involves at most $k(d+1) \le (d-1)(d+1)$ probes. Observe that $G_i$ has $\frac{d}{2}n_i \ge \frac{d}{2}n_1 = d(d+1)$ edges. Hence at least $d+1$ of the edges of $G_i$ are not involved in any of these probes, and consequently there is an edge $(u,w)$ (not involving $v$) that is not probed. Now construct $G_{i+1}$ by including $G_i$ as one of the $H_j$ and the edge $(u,v)$ as the edge removed from this $H_j$. Hence $v$ maintains in $G_{i+1}$ the same set of neighbors that it has in $G_i$, and neither $v$ nor its neighbors change their color under $W_{i+1}$, giving the desired contradiction.
\end{proof}

For every $t$ in the range $0 \le t \le n_{i+1} - n_i$, let $W_{i,t}$ be the algorithm for $F_i$ in which the name space is $[t+1, t+n_i]$ rather than $[n_i]$, and each  queried vertex uses its algorithm from $W_{i+1}$. Under this notation, $W_{i,0} = W_i$. The proof of Lemma~\ref{lem:W_i} shows that algorithm $W_{i,t}$ is a weak 2-coloring algorithm for $F_i$ with name space $[t+1, t+n_i]$,

Lemma~\ref{lem:W_i} together with the inductive hypothesis implies that $k \ge i$. (In fact, it could be that $F_i$ by itself requires more than $i$ probes, and then we are done. But suppose that we are not yet done.) Together with our assumption (for the sake of contradiction) that $k < i+1$ we get that $k = i$. This implies that $W_i$ satisfies the conditions of Lemma~\ref{lem:adversary} with respect to the family $F_i$. Likewise, using Proposition~\ref{pro:renaming} we have that for every $t$ in the range $0 \le t \le n_{i+1} - n_i$, algorithm $W_{i,t}$ satisfies the conditions of Lemma~\ref{lem:adversary} with respect to the family $F_i$ and the name space $[t+1, t+n_i]$.

We now exhibit a graph $G_{i+1} \in F_{i+1}$ on which $W_{i+1}$ makes at least $i+1$ probes, showing that $k \ge i+1$. Recall that $G_{i+1}$ needs to be constructed from $d$ graphs $H_j$ (for $1 \le j \le d$) each from $F_i$, and two auxiliary vertices $a_{i+1}$ and $b_{i+1}$. We use an adversary argument to construct $G_{i+1}$. For every $1 \le j \le d$, the adversary uses the name space $[(j-1)n_i + 1, jn_i]$ for graph $H_j$. The adversary names $a_{i+1}$ as $dn_j + 1$, and makes the $j$th port of $a_{i+1}$ lead into $H_j$. The adversary will use as subroutines the adversaries $A_{F_i,W_{i,t}}$ whose existence is guaranteed by Lemma~\ref{lem:adversary}.

The query that the adversary gives to $W_{i+1}$ is the vertex $a_{i+1} = dn_j + 1$. First, assume that at least one probe is made to a port of $a_{i+1}$, and assume w.l.o.g. that the first probe is such a probe. Upon a probe to port $j$ of vertex $a_{i+1}$ the adversary replies by the vertex $v_j$ that is the query provided by adversary $A_{F_i, W_{i,(j-1)n_i}}$. Likewise, given any probe in graph $H_j$, the adversary relays the probe to the corresponding adversary $A_{F_i, W_{i,(j-1)n_i}}$, and relays the received answer back to $W_{i+1}$. 

Suppose for the sake of contradiction that $W_{i+1}$ stops after $i$ probes (rather than $i+1$), and suppose without loss of generality that it gives $a_{i+1}$ the color~0. As $W_{i+1}$ spent at least one probe on a port of $a_{i+1}$, only $i-1$ probes are left for any of the graphs $H_j$. By Lemma~\ref{lem:adversary}, in each $H_j$ the corresponding adversary $A_{F_i, W_{i,(j-1)n_i}}$ can complete the graph $H_j$ is such a way that  $W_{i,(j-1)n_i}(v_j)$ must be $0$. This gives the graphs $H_1, \ldots, H_d$.

To complete the description of $G_{i+1}$, we need to select in each $H_j$ an edge incident with the corresponding $v_j$ that will be removed, so as to allow $a_{i+1}$ to connect to $v_j$ (and $b_{i+1}$ to connect to the other endpoint of the edge). To ensure that in $G_{i+1}$ we have that $W_{i+1}(v_j) = 0$, the removed edge will be one that was not probed by $W_{i+1}$ in the adversarial process described above, and also not probed by $W_{i,(j-1)n_i}$ upon query $v_j$. Altogether this excludes at most $i-1 + i < d$ ports of $v_j$ (the inequality holds because $i \le q-1$ and $2q < d + 3$), and hence $v_j$ has an incident edge that can be removed. On $G_i$ constructed as above, $a_{i+1}$ and all its neighbors are all colored~0, contradicting that $W_{i+1}$ ends after $i$ probes. Hence $W_{i+1}$ uses at least $i+1$ probes.

Recall that we assumed that at least one probe is made to a port of $a_{i+1}$. If no such probe is made, then the adversary can reply to all of the probes and complete the graphs $H_j$ arbitrarily. Then, for each $j$, the adversary can remove some unprobed edge to a vertex $v_j$ that is colored $0$ (there must be at least one  such vertex and at least one such edge in each $H_j$), and connect $v_j$ to $a_{i+1}$ using this port, and the $j^{th}$ port in $a_{i+1}$.
\end{proof}

\section{Randomized Weak $2$-Coloring}
\label{subsec:rand2}

In this section we  give a randomized  LCA for weak $2$-coloring, \algorithmref{alg:lca2rand}. Each node is randomly assigned a color, and then the inconsistencies created by this coloring are corrected. We consider four possible implementations of the correction process: using strong probes with arbitrary parent selection, the probe complexity of \algorithmref{alg:lca2rand} is $\Theta(\log{n})$; using strong probes with randomized selection, it is $\Theta\left( \frac{\log{n}}{\log\log{n}}\right)$; using weak probes with arbitrary parent selection, it is $\Theta(\log^2{n})$; and using weak probes with randomized selection, it is $\Theta\left( \frac{\log^2{n}}{\log\log{n}}\right)$.

\algorithmref{alg:lca2rand} is the following. First, each vertex $v$ is assigned a color  $c_{temp}(v)$  uniformly at random from $\{0,1\}$. We say that a vertex $v$ is  \emph{good} if it has a neighbor $u$ such that $c_{temp}(v) \neq c_{temp}(u)$. Otherwise, $v$ is \emph{bad}. Every good vertex $v$'s color is fixed to be $c_{temp}(v)$. Bad vertices choose a parent independently of $c_{temp}$. When a vertex $v$ is queried, if it is good, it returns its color, otherwise, it probes its parent. If its parent is good, $v$ returns the complementary color of $\parent(v)$, otherwise it continues iteratively until it encounters either a good vertex or a vertex that it had already encountered. In the first case, we call the encountered vertex a (primary)  \emph{root}; in the second case, this necessarily closes a unique cycle and we call the lowest ID vertex on the cycle a (secondary) \emph{root} and color it $0$. We then color the vertices on this cycle from the root using alternating colors. We differentiate between primary and secondary roots only for clarity; the algorithm treats them identically.

\begin{algorithm}[h]
	\caption{\textsc{Generic Randomized Weak $2$-coloring LCA}}\label{alg:lca2rand}
	\SetKwInOut{Input}{Input}\SetKwInOut{Output}{Output}\SetKwInOut{Inquiry}{Inquiry}
	
	\Input{$G=(V,E)$, query $v$, random function $c_{temp}:V\rightarrow \{0,1\}$}
	\Output{a color in $\{0,1\}$}
	\BlankLine
	
	$G'(v)=(V',E')$, $V'=\emptyset$, $E'=\emptyset$\;
	$u=v$\;

	$V' \leftarrow u$\;
	\While{1}
		{
		\tcc{Step $1$}
		\If{$\exists w \in N(u) : c_{temp}(u) \neq c_{temp}(w)$}
		{$root=u$\;
				$c(u)=c_{temp(u)}$\;
				Break\;
	}\Else{
				\tcc{Step $2$}
			Choose $\parent(u)$ from $u$'s neighbors, independently of $c_{temp}$\;
				\tcc{The method of choosing determines the running time, but any method works}
						$V' \leftarrow \parent(u)$\;
						$E' \leftarrow (u,\parent(u))$\;
			\If{$\parent(u)  \in V'$}
	{Let $root$ be the vertex with the smallest ID in the cycle of $V'$\;
		$c(root)=0$\;
		Break\;}
\Else{

		$u=\parent(u)$\;}
	}}

		\tcc{Step $3$}
		Let $\ell(v)$ be the number of invocations of the ``parent" operations  that are needed in order to reach $root$ when starting at $v$. Then $c(v) = \ell(v) \mod 2$\;
	Return $c(v)$\;

\end{algorithm}

\begin{lemma}
	\algorithmref{alg:lca2rand} produces a weak 2-coloring.
\end{lemma}

\begin{proof}
	Regardless of $c_{temp}$, for every query $v$, the algorithm terminates (because $G'(v)$ can have at most $n$ distinct vertices). Hence every vertex indeed receives a color.
	We show that $v$ must have a neighbor $w$ with $c(v) \not= c(w)$. If $v$ receives its color in Step~$1$ then there is some neighbor $w$ for which $c_{temp}(v) \neq c_{temp}(w)$. $w$ must also receive its color in its own Step~$1$ and $c(v) \not= c(w)$. If $v$ receives its color in Step $2$ or $3$ then there are two cases to consider.

	\begin{enumerate}
		
		\item Vertex $v$ is not the root of $G'(v)$. In this case, consider the parent $w$ of $v$. When $v$ is queried, it assigns to $w$  the color $c(v)-1$. When $w$ is queried, it assigns itself $c(v)-1$ as well. This is because $G'(w)$ and $G'(v)$ are either identical (if $G'(v)$ is a cycle), or $G'(v)$ has an extra edge $(v,w)$. As $v$ is not the root of $G'(v)$, $G'(v)$ and $G'(w)$ have the same root. Hence $c(v) \not= c(w)$.
		
		\item Vertex $v$ is the root of $G'(v)$. Then $v$ sits on the unique cycle of $G'(v)$. Let $w$ be the vertex on that cycle such that $v$ is the parent of $w$. Then necessarily $G'(w)$ is the same cycle with $v$ as its root, and hence $c(v) \not= c(w)$.
			\end{enumerate}
	\end{proof}

\subsection{Strong Probes}
\label{sub:strongp}
\paragraph{Arbitrary choice of parent.}

We first analyze \algorithmref{alg:lca2rand} for arbitrary strong probes.

\begin{thm}\label{thm:randomf}
	For every $n$ vertex graph and every $1 \le k \le n$, the probability (over the random choice of $c_{temp}$) that there is a vertex $v$ that  uses more than $k + 1$ strong probes, when executing ~\algorithmref{alg:lca2rand} with an arbitrary (but independent of $c_{temp}$) choice of parent for each vertex,  is at most $2^{1-k}n $. In particular, for $k = (\alpha+1) \log n$  this probability is  at most $n^{-\alpha}$.
\end{thm}

  \begin{proof}
  	Denote the subgraph created by taking only edges from vertices to their parents by $G'=(V',E')$. For every $1 \le k \le n$, there are at most $n$ directed (simple) paths in $G'$ of length $k$: for every $v \in V'$, there is at most one path of length $k$ from $v$, as the out-degree of each vertex in $V'$ is exactly $1$.
  	For any set of vertices $S \subseteq V'$,   the probability that $S$ is monochromatic with respect to $c_{temp}$ is $2^{1-|S|}$. Applying the union bound over the $n$ possible paths of length $|S|$ completes the proof.
  \end{proof}

The following proposition shows that the upper bound of Theorem~\ref{thm:randomf} is tight, even for graphs of diameter $2$. 

\begin{proposition}
	\label{pro:diameter3}
	There are graphs of diameter~$2$ for which the longest monochromatic path in the lowest ID scheme is of length $\Omega(\log n)$ with probability $1-o(1)$ over the random choice of $c_{temp}$.
\end{proposition}

\begin{proof}
	Let $G$ be a graph as in the proof of~\propref{pro:diameter2}: a path with $n-1$ vertices with IDs from $n-1$ to~$1$ and  one vertex with ID $n$ connected to all other vertices. The probability that a path of length $c\log{n}$ is monochromatic is $1/n^c$. There are $\frac{n}{c\log{n}}$ such disjoint paths. Therefore, the probability that none of these is monochromatic is at most $(1-\frac{1}{n^c})^{n/(c\log{n})} = o(1)$, for $c<1$.
\end{proof}

\paragraph{Random choice of parent.}

We now show that the random lowest ID scheme for choosing parents does strictly better.
We say that a vertex is {\em boring} if it has the same color as all its neighbors under $c_{temp}$. Namely, the boring vertices are those that are not secondary roots. Let $B(c,G)$ denote the size of the largest boring connected subgraph in $G$ under a 2-coloring $c_{temp}$.

\begin{lemma}
\label{lem:randomf}
For every $n$ vertex graph $G$, the probability (over the random choice of $c_{temp}$) that $B(c_{temp},G) \ge k$ is at most $\min_{0 < s < k}\left\{\frac{n^2}{{k \choose s}}\right\}$.
\end{lemma}

\begin{proof}
Consider a connected subgraph $S$ of size $s$ and let $t$ denote the number of vertices not in $S$ but adjacent to $S$. For $S$ to be boring all these $t+s$ vertices need to be of the same color, and this happens with probability $2^{1 - s - t}$. In any $n$ vertex graph $G$, the number of ways of picking a connected subgraph of $s$ vertices that has exactly $t$ neighbors is at most $n{s + t \choose s} \le 2^{t+s-1}n$ (see \appendixref{sec:conn}). Hence the expected number of boring connected subgraphs of size $s$ is as most $\sum_{t=0}^{n-1} 2^{t+s-1}n2^{1 - s - t} \le n^2$.

Suppose that the size of the largest boring connected subgraph is $k$. Then for every $s < k$, this subgraph has ${k \choose s}$ subgraphs of its own that are boring. By Markov's inequality, this event  happens with probability at most $\frac{n^2}{{k \choose s}}$ (over the choice $c_{temp}$). 
\end{proof}

\begin{thm}\label{thm:r2d2}
For every graph, the longest monochromatic path has length $O\left(\frac{\log n}{\log\log n}\right)$ w.h.p. in  the random lowest neighbor scheme.
\end{thm}
\begin{proof}
Fixing $s = \frac{k}{2}$ in Lemma~\ref{lem:randomf}  implies that after the random coloring, the size of the largest boring component is at most $8\log n$, with high probability. Moreover, there are at most $n$ components in the graph.

Consider now an arbitrary component $C$ of size $k \le 8\log n$ (but still $k = \Omega\left( \frac{\log{n}}{\log\log{n}}\right) $). We bound from above the probability that under the random lowest neighbor scheme, $C$ contains a monochromatic path of length $2t$ (for simplicity, assume $t$ divides $k$). Consider an arbitrary path $P$ in $C$ (not directed, prior to choosing parents) of length $2t$. For vertex $i$ ($1 \le i \le 2t$) on $P$ let $n_i$ denote the number of {\em fresh} neighbors of $i$ in $C$, where a neighbor of $i$ is fresh if it was not a neighbor of any vertex $j < i$ where $i-j$ is even. Hence $\sum_{\mbox{odd}\; i} n_i \le k$ and $\sum_{\mbox{even}\; i} n_i \le k$. By Lemma~\ref{lem:nonincrease2}, in any directed (nonincreasing) path that agrees with $P$ up to vertex $i$, vertex $i$ has only $n_i$ choices for a parent (if $i$ chooses as parent a non-fresh neighbor, some ancestor of $i$ would have already chosen it). Consequently, the number of candidate monochromatic paths of length $2t$ is at most:
$$k{k+t \choose t}^2 \prod_{i=1}^{2t} n_i \leq k{k+t \choose t}^2 \left(\frac{k}{t}\right)^{2t} \le e^t \left(\frac{k}{t}\right)^{5t},$$
where $k$ is the number of starting vertices, ${k+t \choose t}$ is the number of ways of choosing the sequence of $n_i$ in odd (or even) locations, and the number of choices of parents is maximized when all $n_i$ are the same.

The probability that a particular candidate of path $P$ is realized as a monochromatic path in the random neighbor scheme is at most $\left(\frac{1}{t!}\right)^2$. Therefore, the probability that there is at least one such path of length $2t$ is at most
$\left(\frac{n}{(t!)^2}\right)e^t \left(\frac{k}{t}\right)^{5t}$.
Setting $t= \alpha\frac{\log n}{\log\log n}$ for some $\alpha>1$ gives that (as $k\leq 8 \log{n}$)
 with high probability over the choice of random IDs the longest monochromatic path in $G$ has length $O\left(\frac{\log n}{\log\log n}\right)$.
\end{proof}

\begin{proposition}
	There are graphs for which the longest  monochromatic path has length $\Omega\left(\frac{\log n}{\log\log n}\right)$ in the random lowest neighbor scheme  with  probability $1-o(1)$ (over random IDs and random colors).
\end{proposition}

\begin{proof}
	Consider the cycle. The probability that a segment of length $t$ is both monochromatic and monotone decreasing is $\frac{1}{2^tt!}$, and there are $\lfloor n/t \rfloor$ such disjoint segments. The probability that none of these segments satisfies both properties is at most $\left( 1-\frac{1}{2^tt!}\right)^{n/t}$. Setting $t=\Omega\left(\frac{\log n}{\log\log n}\right)$ completes the proof.
\end{proof}


\subsection{Weak Probes}
\label{sub:weakp}
The only difference between using strong and weak probes in the implementation of~\algorithmref{alg:lca2rand} is that one strong probe suffices for Step $1$, but we may need $d(u)$ weak probes to check whether a neighbor of $u$ is colored differently than $u$.

\begin{corollary}[to~\theoremref{thm:randomf}]\label{thm:randoml}
~\algorithmref{alg:lca2rand} uses $O(\log^2{n})$ weak probes w.h.p. when implemented with an arbitrary (but independent of $c_{temp}$) choice of parent. 
\end{corollary} 

\begin{proof} When implementing~\algorithmref{alg:lca2rand} with weak probes, Step~1 can take as many probes as the degree of the current vertex (denoted $u$ in the pseudocode). Note, however, that the probability that none of the neighbors of $u$ have a different color decays exponentially with the degree. Hence, if the degree is $b \log{n}$, for some $b>0$, then the probability that none of the neighbors is colored differently is $n^{-b}$. By~\theoremref{thm:randomf}, we know that the length of any monochromatic path is at most $(\alpha+1)\log{n}$ w.h.p. If any of the vertices on such a  path  has more than $\alpha \log{n}$ neighbors, we are done (for failure probability $n^{-\alpha}$). Otherwise, we need to make at most $\alpha(\alpha+1)\log^2{n}$ probes.
	\end{proof}

\begin{proposition}
	\label{pro:low}
	There are graphs for which~\algorithmref{alg:lca2rand} uses   $\Omega(\log^2 n)$ weak arbitrary  probes with constant probability (over the random choice of $c_{temp}$).
		\end{proposition}

\begin{proof}
	Let $0<c<1$ be a constant, and let $G$ be a collection of $t=n/c\log{n}$ disjoint cliques of size $c\log{n}$, denoted $C_1, \ldots, C_t$. The probability that at least one such clique is monochromatic is at least $1/4$, for sufficiently large $n$. (If we wish to have a connected graph as a lower bound, we can simply connect  $C_i$  to $C_{i+1}$ by a single edge for $i \in [t-1]$, with minimal effect.)
	Let $E' = \Omega(\log^2{n})$ be the number of edges in each clique.	For deterministic choice of parent, it is easy to guarantee that for all cliques, there is a particular query for which  a cycle will be closed only after $E'-\log{n}$ probes (i.e., only when the last vertex of the clique is reached).
	\end{proof}

\begin{corollary}[to~\theoremref{thm:r2d2}]\label{coro:ra}
~\algorithmref{alg:lca2rand} uses $O\left( \frac{\log^2{n}}{\log\log{n}}\right) $ weak probes w.h.p. when implemented using the lowest random neighbor scheme. (The probability is over the random choices of parent and $c_{temp}$).
\end{corollary} 

\begin{proof} The proof is similar to the proof of ~\corollaryref{thm:randoml}: By~\theoremref{thm:r2d2}, we know that the length of any monochromatic path is at most $\frac{\gamma\log{n}}{\log\log{n}}$, for some $\gamma>0$. If any of the vertices on such a  path  has more than $\alpha \log{n}$ neighbors, for an appropriately chosen $\alpha>0$, we are done. Otherwise, we need to make at most $\frac{\alpha\gamma\log^2{n}}{\log\log{n}}$ probes w.h.p.
\end{proof}

	\begin{proposition} 
	There are graphs for which~\algorithmref{alg:lca2rand}, when implemented using the random lowest neighbor scheme, uses  $\Omega(\log^{2} {n}/\log\log{n})$  weak probes with constant probability over the random choice of $c_{temp}$.
\end{proposition}
\begin{proof}
	Let $0<c_1, c_2<1$ be  constants. Let $G$ be a collection of $t=\Theta(n/\log{n})$ copies of the following graph $G_i$: a path of length $r=\frac{c_1\log{n}}{\log\log{n}}$, $v_i^1, \ldots, v_i^r$,  where all $v_i^j$ are connected to (all vertices of a) clique of size $k = c_2\log{n}$.  The probability that $G_i$ is monochromatic is $2^{-r-k+1}$. The probability that for all $j \in [r-1]$, the parent of $v_i^j$ is $v_i^{j+1}$, is $\frac{k!}{(k+r)!} \ge (k+r)^{-r}$. For $k$ and $r$ as above, the probability that both events happen is at least roughly (up to low order multiplicative terms) $2^{-c_2\log n}2^{-c_1\log n}$. 
Letting $\alpha=2(c_1+c_2)$ and setting $c_1, c_2$ to be appropriately small,  the probability that none of the $G_i$ have this property is $ \left( 1 - \frac{1}{{n^{\alpha}}}\right) ^{\Omega(n/\log{n})} = o(1)$.
\end{proof}

\subsection{Seed Length}

Thus far in this section, we assumed that we have a perfectly random string of unlimited length, that we can use to generate $c_{temp}$. We  show to to implement~\algorithmref{alg:lca2rand} using a seed of logarithmic length. Our result, whose proof we defer to~\appendixref{app:kwise}, is the following. 

\begin{theorem}\label{thm:ak}
	\algorithmref{alg:lca2rand} is a randomized LCA for weak $2$-coloring arbitrary graphs that can be implemented as  	\begin{enumerate}
		\item  an $\left( O(\log^2{n}), O(\log{n}), \frac{1}{\poly{n}}\right) \text{ - WP LCA, or }$
		\item an
 $\left( O(\log{n}), O(\log{n}), \frac{1}{\poly{n}}\right)  \text{ - SP LCA}.$	
	\end{enumerate}
\end{theorem}

\subsection{Discussion}

 We showed that for $d=O\left( \frac{\log{n}}{\log\log{n}}\right)$, a linear dependence on the degree is necessary for deterministic weak $2$-coloring LCAs. It would be interesting to resolve whether such a dependence is also necessary for higher degrees. We conjecture that this is indeed the case. If so, it would give a separation between deterministic and randomized LCAs for this problem. 
 
\section{Lower Bound for Vertex Cover}
\label{sec:VC}

In this section we prove Theorem~\ref{thm:lower}:
 {
	\renewcommand{\thetheorem}{\ref{thm:lower}}
 \begin{theorem}
 		For any $\eps<\frac{1}{2}$, any randomized SP LCA that computes a vertex cover  whose size is a $(\lowerapp)$-approximation to the size of the  minimal vertex cover requires at least $\lowerprobes$ probes.
 \end{theorem}
 	\addtocounter{theorem}{-1}
 }


In order to prove \theoremref{thm:lower} , we use the minimax theorem of Yao~\cite{Yao77},
by showing a lower bound on the expected number of queries required by a deterministic LCA when the 
input is selected from a certain distribution.
To this end we construct, for infinitely many values of $n$, a family of graphs, 
parametrized by a constant $k \geq 3$. 

Let $p$ be a prime number, and $\Z(p)$ its associated field.
Let $G^*=(U^* \cup W^*, E)$ be a bipartite graph, where $|U^*|=p^k$, $|W^*| = p^2$. Label each vertex in $U^*$ by a $k$-tuple $(a_0, a_1, \ldots, a_{k-1})$, $a_i \in \Z(p), i \in \{0,1, \ldots, k-1\}$ and the vertices in $W^*$ by a pair $(b_0, b_1)$, $b_i \in \Z(p), i \in \{0,1\}$.
Associate each vertex $u_j=(a_0, a_1, \ldots,  a_{k-1}) \in U^*$ with the polynomial $f_j(x)=a_0+a_1 x+\cdots +a_{k-1} x^{k-1}$.
Connect $(b_0, b_1) \in W^*$ to $u_j$ iff $b_1=f_j(b_0)$.

\begin{claim}\label{claim:kdiff}
	For every two different vertices $u_i, u_j \in U^*$, $N(u_i)\cap N(u_j) \leq k-1$.
	\end{claim}
	\begin{proof}
	 Consider $f_i$ and $f_j$, the polynomials associated with $u_i$ and $u_j$ respectively. Let $g(x)=f_i(x)-f_j(x)$. As $g$ is not the zero polynomial, it has at most $k-1$ roots in $\Z(p)$.
	\end{proof}

Let $\mathbb{G}$ to be a graph consisting of two identical copies of $G^*$. We denote these two copies by $G'=(U'\cup W',E')$ and $G''=(U''\cup W'',E'')$. Let $U=U' \cup U''$; let $W=W'' \cup W''$; let $n=|U \cup W| =2(p^k+p^2)$. 

We define the following operation on $\mathbb{G}$ (see \figref{fig1}).  Let $e'$ and $e''$ be edges such that 
 $e'=(u', w'): u' \in U', w' \in W'$ and $e''=(u'', w''):u'' \in U'', w'' \in W''$. Remove $e'$ and $e''$ from $\mathbb{G}$, and add an edge $e=(u', u'')$. 
We call this operation $\fuse(\mathbb{G},e',e'')$, and call $u$ and $u'$ the \emph{fusion 
	vertices}. Note that there are $p^{k+1}$ possible choices for $e'$ and  $p^{k+1}$ possible choices for $e''$. 

	   \begin{figure}[t]
	   	\centering
	   	\includegraphics[width=14cm]{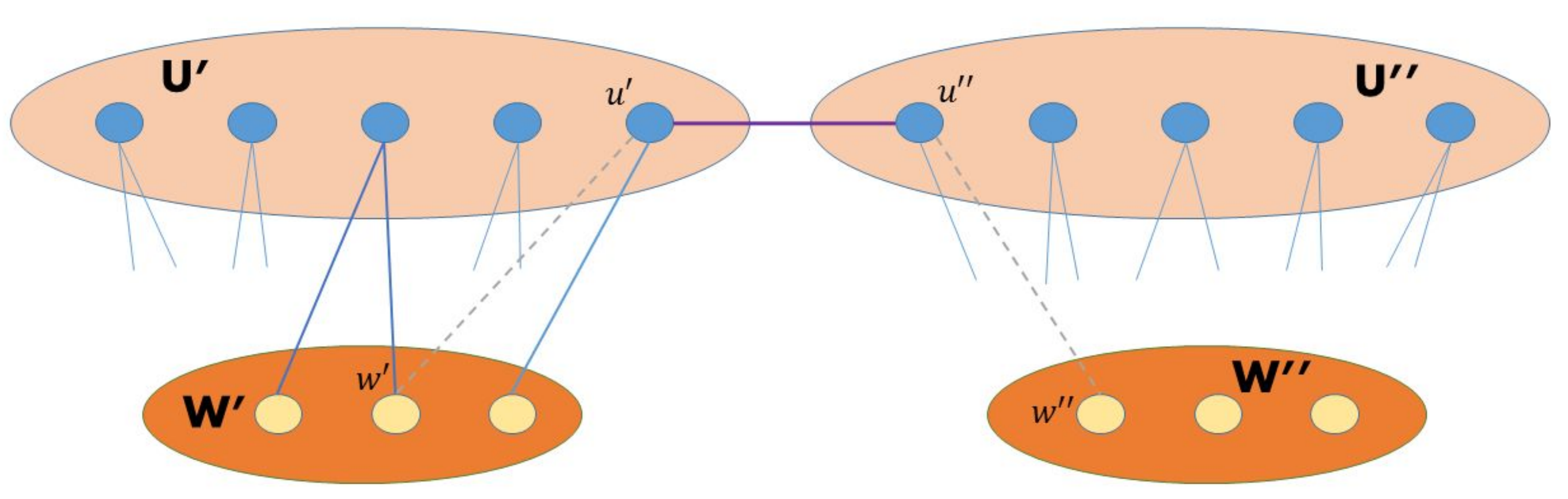}
	   	\caption{The graph  $\fuse(\mathbb{G},(u',w'),(u'',w''))$. The dashed edges $e'=(u',w')$ and $e''=(u'',v'')$ have been removed and the edge $(u',u'')$ has been added.}
	   	\label{fig1}
	   \end{figure}

Given a graph $G=\fuse(\mathbb{G},(u',w'),(u'',w''))$, the optimal size of the vertex cover of $G$ is at most $2p^2+1$, as $W$ and a fusion vertex constitute a vertex cover. 

Note that a vertex can locally detect whether it is in $U$ or in $W$ just by 
looking at its own degree. 
However, to detect whether it is one of the fusion vertices, it needs
to determine the degrees of its neighbors, which is impossible to do with number of probes
significantly smaller than its degree.\footnote{Compare this with a  distributed algorithm in the LOCAL model, for which one round suffices to determine this.}
%

We now describe the graph family we use.
Let $\mathbb{G}=(U \cup W, E)$ be as above. Let $\Pi$ be the set of all possible namings of $U \cup W$ by the ID set $[n]$. 
Let $T=E'\times E''$. 
Given a naming $\pi\in\Pi$ and a pair of edges $\tau=(e',e'')\in T$, the  graph $\mathbb{G}(\tau,\pi)$ is defined
as follows: 
\begin{enumerate}
	\item 
	The topology of $\mathbb{G}(\tau,\pi)$ is given by $\fuse(\mathbb{G},e',e'')$. 
	\item
	The vertices of $\mathbb{G}(\tau,\pi)$ are named  according to $\pi$.
\end{enumerate}
The family of graphs we consider is $\calG(\Pi, T)=\{\mathbb{G}(\tau,\pi)\mid\tau\in T,\pi\in\Pi\}$.  
We now analyze the behavior of a given deterministic LCA $A$ with proble complexity less than
$p/k$ on a graph chosen uniformly
at random  from
$\calG(\Pi, T)$.  
We first make the following simplification.
Suppose that
 $A$ running on some  $\mathbb{G}(\tau,\pi)$ is given $v$ as a query.
If $A$ probes $w'_{\tau}$ or $w''_{\tau}$, it knows that $v$  is neither $u'_{\tau}$ nor $u''_{\tau}$ and hence $A$ need not  include $v$ in the VC. For simplicity, 
we also assume that if $A$ probes $u'_{\tau}$ or $u''_{\tau}$, it knows not to add $v$ to the VC. Note that this only strengthens $A$, hence we can make this assumption without loss of generality.

We use the following definition.
\begin{definition}[View]
	Let $\mcA$ be a deterministic SP LCA. We denote by $\view(\mcA,G, v)$ the subgraph that $\mcA$ discovers when queried on vertex $v$ in graph $G$, i.e., the set of all probed vertices and their neighbors. 
\end{definition}


\begin{claim}
	\label{claim:miss}
	Let $\mcA$ be an SP LCA with probe complexity less than $p/k$.
	Let $G\in \calG(\Pi, T)$. 
	Assume that $A$ is queried on 
	$G$ with vertex $v$. Then there is some vertex $w \in N(v)$ such that $w$ has no neighbors except $v$ in $\view(\mcA, G, v)$. (That is, neither $w$ nor any of its neighbors (except $v$) is probed in $\mcA(G,v)$.)
\end{claim}
\begin{proof}
	$\mcA$ probes $v$ and, say, $a$ vertices from $U$ and $b$ vertices from $W$, for some $a, b \in \N$ such that $a+b < p/k$. From \claimref{claim:kdiff}, $A$ sees at most $a(k-1)$ vertices from $N(v)$ as a result of probing vertices in $U \setminus \{v\}$. Furthermore, $A$ sees at most $b$ vertices from $N(v)$ as a result of probing vertices in $W$. As $a(k-1)+b < p = |N(v)|$, at least one vertex in $N(v)$ is only seen once, while probing $v$.	
\end{proof}


\noindent
 Consider an input graph $\mathbb{G}(\tau,\pi)$, where $\tau=((u'_{\tau},w'_{\tau}),(u''_{\tau},w''_{\tau}))$.
 We use the following notation.
 \begin{compactitem}
 	\item $A_v$ denotes the event that $\mcA$ is given $v$ as a query.
 	Note that $A_v$ is independent of $\pi$ and $\tau$. 
 	\item $X_{\tau, \pi, i}$ denotes the event that none of $u'_{\tau}, w'_{\tau}, u''_{\tau}, w''_{\tau}$ is probed when $A$ is queried on $i \in [n]$.
 \end{compactitem}
%
\begin{claim}\label{claim:view}
	Fix $\pi$ and $\tau$. 
	Let $v$ be a vertex in $U' \setminus \{u'_{\tau}\}$. If $A_v$ and $X_{\tau, \pi, \pi(v)}$ hold, there exist  $\pi_1\in\Pi, \tau_1\in T$ such that
	$\view(\mcA, \mathbb{G}( \pi_1, \tau_1), u'_{\tau_1})=\view(\mcA, 	\mathbb{G}(\pi, \tau), v).$
\end{claim}

\begin{proof}
	By \claimref{claim:miss}, there is some vertex $w \in N(v)$ that has no  neighbors other than $v$ in $\view(\mcA, G(\pi, \tau), v)$.
	Let $\pi_1$ be identical to $\pi$ except that $v$ and $u'_{\tau}$ are interchanged.
 Set $\tau_1 = ((v, w), e''_{\tau})$. The claim follows. 
\end{proof}
A symmetrical argument holds for $v \in U'' \setminus \{u''_{\tau}\}$.

\begin{claim}
	\label{claim:doesnt}
	Fix $\mathbb G$, and let  $i  \in [n]$ be the ID of the vertex given to  a deterministic SP LCA $\mcA$ as a query. 
	If the probe complexity of $\mcA$ is less than $p/k$, then $\Pr[X_{\tau, \pi, i}]\geq 1-\frac{1}{kp}$,
	where the probability is over the choice of $\pi$ and $\tau$.
\end{claim}
\begin{proof}
	Fix $\pi$. 
If $A$ probes $a$  vertices in $U$ and $b$ vertices in $W$, it will hit one of $u'_{\tau}, w'_{\tau}, u''_{\tau}, w''_{\tau}$ with probability at most
$\frac{a}{p^k} + \frac{b}{p^2}$,
over the choice of $\tau$. Since $a,b\ge 0$ and $a+b\le p/k$, the probability is maximized for $a=0, b=p/k$. As this bound holds for any $\pi$, the result follows.
\end{proof}

 \claimref{claim:view} and \claimref{claim:doesnt} imply that when
 a deterministic SP LCA  is queried on a vertex  $u\in U$ from a random graph
 in $\calG(T,\Pi)$, it cannot discern in less than $p/k$ probes whether
 $u$ is a fusion vertex with probability greater
 than $\frac{1}{kp}$. Hence the LCA must add vertex $u$ to the VC, because at least one fusion
 vertex \emph{must} be in the VC. Therefore, the size of the VC that $A$ computes is at least $p^{k}-O(1)$, whereas the optimal VC has size at most $p^2+1$:
 
 \begin{theorem}\label{thm:det}
 	There does not exist a deterministic  SP LCA $A$  that computes a VC that is an 
 	$(\lowerapp)$-approximation to  the optimal VC and uses fewer than $\lowerprobes$ probes with probability
 	greater than $\frac{1}{kp}$  on a graph chosen uniformly at random from $\calG(\Pi, T)$.
 \end{theorem}
 \begin{proof}
 	Let $\epsilon=1/k$. Recall that $n=\Theta(p^k)$. We have shown that
 	if the number of probes is less than $p/k=\Theta(n^\epsilon\cdot\epsilon)$, then
 	the approximation ratio is at least $p^k/p^2-o(1)=n^{1-{2/k}}-o(1)$. 
 \end{proof}
 
Applying Yao's principle \cite{Yao77} to Theorem~\ref{thm:det} completes the proof of Theorem~\ref{thm:lower}.
 
 \begin{corollary}
 	Any SP LCA for maximal matching on arbitrary graphs requires $\Omega({n}^{1/2-o(1)})$ probes.
 	\end{corollary}
\begin{proof} 
	It is well known that, given any maximal matching, taking both end vertices of every edge gives a $2$-approximation to the VC (e.g., \cite{Vaz01}).
	Therefore, an LCA for maximal matching would immediately give a $2$-approximation to the minimal vertex cover. Setting $2=\Theta(n^{1-2\epsilon})$ in Theorem~\ref{thm:lower} gives $\eps = 1/2$. The result follows.
	\end{proof}


\section{Almost Maximum Matching in Regular Graphs}
\label{sec:match1}

Our goal in this section is the following: given a regular input graph $G=(V,E)$, the LCA receives an edge $e \in E$ as a query;
for some almost maximum matching  $M \subseteq E$, if $u \in M$, the LCA outputs ``yes'', otherwise it outputs ``no''.  All outputs must be consistent with the same $M$.
\subsection{Bounded Degree Graphs}
We give an LCA for finding a $(1-\eps)$ approximation to maximum matching on bounded degree graphs whose probe complexity is independent of $n$.
We use the following result of~\cite{YYI12} (rephrased), that can be found in the proof of Theorem 3.7 in~\cite{YYI12}.
\begin{lemma}[\cite{YYI12}]\label{lem:yyi}
	Let $G=(V,E)$ be a graph whose degree is bounded by $d$. There exists a randomized LCA for $(1-\eps_1)$-approximate maximum matching, whose probe complexity is $d^{O\left( \frac{1}{\eps_1^2}\right)}$ in expectation, where expectation is taken both other the choice of queried edge and over the randomness of the LCA.
\end{lemma}

For completeness, we summarize this algorithm in~\appendixref{app:bounded}.
\begin{theorem}\label{thm:bounded}
	There exists an SP LCA that finds a $(1-\eps)$-approximate maximum matching (in expectation) on graphs with degree bounded by $d$ that uses  $\left( d + \frac{1}{\eps}\right) ^{O\left( \frac{1}{\eps^2}\right)}$ probes per query.
	\end{theorem}

(Note that in~\theoremref{thm:bounded} the matching is a $(1-\eps )$ approximation in expectation, while the probe complexity is worst case; the opposite is true for~\lemmaref{lem:yyi}.)


\begin{proof}[Proof of~\theoremref{thm:bounded}]
	Set $\eps_1=\eps_2=\eps/2$. Consider the LCA of~\lemmaref{lem:yyi} with parameter $\eps_1$:
	For any edge $e \in E$, let $X_e$ be the random variable denoting the number of probes required to reply to the query for $e$. By~\lemmaref{lem:yyi}, $\expect[X_e] \leq d^{\left( \frac{c}{\eps_1^2}\right)}$, for some constant $c>0$, where expectation is taken both other the choice of queried edge and over the randomness of the LCA. Let us use the notation $\tau = d^{\left( \frac{c}{\eps_1^2}\right)}$.
	
		Our algorithm is simply the algorithm of~\cite{YYI12} (with parameter $\eps_1$), except that if the algorithm uses more than $\frac{\tau d}{\eps_2}$ probes it returns ``no''; i.e., the edge is not in the approximate maximum matching.

By Markov's inequality, the expected fraction of edges for which $X_e>\frac{\tau d}{\eps_2}$ is at most $\frac{\eps_2}{d}$.
Hence in expectation, the LCA will reply incorrectly on at most $\eps_2 n$ edges; the approximation ratio of the LCA is therefore $(1-(\eps_1+\eps_2))$.
\end{proof}

\subsection{High Degree Regular Graphs}

Let $\epsilon > 0$ be a fixed constant. For some $d_{\epsilon}$ that may depend on $\epsilon$, and for some $d \ge d_{\epsilon}$ (for our proofs $d_{\epsilon} \simeq \frac{1}{\epsilon^2}$ suffices), let $G(V,E)$ be a $d$-regular graph on $n$ vertices. We wish to design a randomized SP LCA that 
outputs a matching of size $(1 - \epsilon)\frac{n}{2}$ in $G$. As $G$ cannot have a matching larger than $n/2$, this provides a $(1 - \epsilon)$-approximation. 
An $\eps$-approximately-$d$-regular graph is one for which the average degree is $d$ and all vertex degrees lie within the range $\left[ (1-\eps)d, (1+\eps)d\right] $.

 {
	\renewcommand{\thetheorem}{\ref{thm:dregularbest}}
\begin{theorem}
	There exists an SP LCA that finds a $(1-\eps)$-approximate maximum matching in expectation on $\eps$-approximately-$d$-regular graphs that uses $\left( \frac{1}{\eps}\right)^{O\left( \frac{1}{\eps^2}\right)}$ probes per query.
\end{theorem}
 	\addtocounter{theorem}{-1}
 }



Our approach is based on {\em sparsification} of $G$. Let $q$ be a parameter that depends on $\epsilon$ (in this section we shall take $q \simeq \epsilon^{-2}$, but in~\appendixref{sec:mmm} we show that $q \simeq \frac{1}{\epsilon}$ also suffices). Construct from $G(V,E)$ a random graph $G'(V',E')$ in two steps:

\begin{enumerate}

\item Initially $V' = V$, and every edge $e \in E$ is included in $E'$ independently with probability $\frac{q}{2d}$. We refer to these edges as {\em surviving} edges.

\item To obtain $V'$, (simultaneously) drop from $V$ every vertex that has more than $q$ surviving edges.

\end{enumerate}

\begin{proposition}
$G'$ has maximum degree $q$.
\end{proposition}

\begin{proof}
By step~2 of the sparsification.
\end{proof}

\begin{lemma}
\label{lem:simulation}
Every strong probe in $G'$ can be implemented by $q+1$ strong probes to $G$.
\end{lemma}

\begin{proof}
On a probe to vertex $v \in V$, the reply needs to be the list of its neighbors in $G'$. Probe $v$ in $G$ so as to get the list of its neighbors in $G$, and hence of its surviving edges. If $v$ has more than $q$ surviving edges, then reply that $v$ is not in $G'$. If $v$ has at most $q$ surviving edges, probe each of their endpoints and include the respective vertex in the reply only if it is a member of $G'$ (has at most $q$ surviving edges).
\end{proof}

The following lemma suffices for proving~\theoremref{thm:dregularbest}. Stronger bounds (of the form $\frac{n}{2} - O(\frac{n}{q})$) are provided by ~\theoremref{thm:sparsify} in~\appendixref{sec:mmm}.

\begin{lemma}
\label{lem:sparsifyq}
For every $G$, the expected size of a maximum matching in graph $G'$ is at least $\left(1 - \frac{1}{\sqrt{q}}\right)\frac{n}{2}$.
\end{lemma}

\begin{proof}
Consider first only the first step of the sparsification. Let every vertex choose at random a color of either~0 or~1 uniformly and independently. Call an edge {\em eligible} only if it is a surviving edge whose endpoints have different colors. For a vertex $w$, let $d_w$ denote the number of eligible edges incident with it. Consider a fractional matching in which the weight of every edge $(u,v)$ is $\frac{1}{\max\left\lbrace d_u,d_v\right\rbrace }$.
\begin{claim}The expected weight contributed by any edge $e \in E$ to the fractional matching is at least $\frac{1}{d}\left( 1 - \frac{1}{\sqrt{q}}\right) $.\end{claim}
\begin{proof} Edge $e$ survives with probability $\frac{q}{2d}$, and, conditioned on surviving, is eligible with probability $\frac{1}{2}$. The expected degree of a vertex in $G'$ is $q/2$, and its expected eligible degree is $q/4$. As its degree is dominated by the binomial distribution, its variance is at most $q/4$. Its standard deviation is therefore at most $\sqrt{q}/2$. 
	From ~\claimref{claim:expectmax},
	$\expect\left[ \frac{1}{max\{d_u, d_v\}}\right] \geq \frac{1}{\frac{q}{4}+\frac{3\sqrt{q}}{2}}.$
	
	 Hence the total expected fractional weight is
	 	\begin{align}\frac{q}{4d}\frac{1}{\frac{q}{4} +\frac{3\sqrt{q}}{2}} &= \frac{1}{d}\left( \frac{1}{1 +\frac{6}{\sqrt{q}}}\right) \notag\\
	 	&\geq \frac{1}{d}\left( 1-\frac{1}{\sqrt{q}}\right) \label{eq:1}.
	 	\end{align}
\end{proof}


The fractional matching above is supported on edges of a bipartite graph, and hence it can be transformed into an integer matching of the same size.

Finally, some of the matched edges might be removed due to vertices that are removed in the second step of the sparsification. However, the total number of vertices removed it that step is of order $n2^{-\Omega(q)}$; $2^{-\Omega(q)}$ is absorbed into Inequality~\eqref{eq:1}. (This last statement requires $q$ to be at least some sufficiently large constant. Such a choice of $q$ can be enforced in our sparsification procedure. Its effect on the upper bound on the number of probes can be absorbed into the $O$ notation in~\theoremref{thm:dregularbest}.)
\end{proof}


The combination of \lemmaref{lem:simulation} and \lemmaref{lem:sparsifyq} implies that we can randomly reduce the problem of approximating maximum matching in a $d$-regular graph $G$ with arbitrarily large $d \ge q$ to that of approximating maximum matching in a graph $G'$ of degree bounded by $q$, with an expected additive loss of at most $\frac{1}{\sqrt{q}}$ in the approximation factor, and a multiplicative loss of $q$ in the number of strong probes. 

\begin{thm}\label{thm:reduction}
There is a randomized reduction from the problem of approximating maximum matching in a $d$-regular graph $G$ with arbitrarily large $d \ge q$ to that of approximating maximum matching in a bipartite graph $G'$ of degree at most $q$, with an expected additive loss of at most $\frac{1}{\sqrt{q}}$ in the approximation factor, and a multiplicative loss of $q$ in the number of strong probes.
\end{thm}

\begin{proof}[Proof of~\theoremref{thm:dregularbest}]
	Let $G$ be a $d$-regular graph. Let $\eps'=\eps''=\frac{\eps}{2}$. Set $d_{\eps} = \frac{1}{\eps'^2}$. If $d>d_\eps$,
combining~\theoremref{thm:bounded} and~\theoremref{thm:reduction} gives required algorithm: use the sparsification technique to obtain $G'$, run the LCA of~\theoremref{thm:bounded} to obtain a $(1-\eps'')$-approximate maximum matching on $G'$, which is also a $(1-\eps)$-approximate matching in $G$.
If $d \leq d_\eps$, the LCA of~\theoremref{thm:bounded} suffices.

This gives a probe complexity of $O\left( \frac{1}{\eps^2}\right) ^{\frac{1}{\eps^2}}$.
To obtain the actual bound of~\theoremref{thm:dregularbest} , set $d_\eps=\frac{1}{\eps'}$ and use~\corollaryref{cor:fin} in lieu of~\theoremref{thm:reduction}.
\end{proof}


\section{Almost Maximum Matching in Regular Graphs with High Girth}

\label{sec:match}

In this section, we prove~\theoremref{thm:girth}. 

 {
	\renewcommand{\thetheorem}{\ref{thm:girth}}
\begin{theorem}
For any $\epsilon > 0$, let $G$ be an $n$-vertex $d$-regular graph of girth $g$, with $d \ge \frac{1}{\epsilon}$ and $g \ge \frac{1}{\epsilon^3}$. Then there is a randomized SP LCA that uses at most $O\left( 1/{\epsilon^7}\right) $  probes and finds a matching that matches at least $(1 - \epsilon)n$ vertices in expectation.
\end{theorem}
 	\addtocounter{theorem}{-1}
}



Our starting point is the following result of Gamarnik and Goldberg~\cite{GG10} (a restatement of their Theorem 5):

\begin{thm}
\label{thm:gamarnik}\cite{GG10}
Let $G$ be a an $d$-regular graph, of girth $g$, where $d \geq 3, g \geq 4$. Then the expected size of the matching found by the randomized greedy algorithm is at least
$$\left(\frac{1 - (d - 1)^{-\frac{d}{d - 2}}}{2} - \frac{d(d - 1)^{\lfloor \frac{g - 2}{2} \rfloor}}{\lfloor \frac{g - 2}{2} \rfloor !}\right)n.$$
\end{thm}

We  use the following corollary. 

\begin{corollary}
\label{cor:gamarnik1}
For every $\epsilon > 0$, if $d \geq \max\{\frac{2}{\epsilon}, 3\}$ and $g \geq 8d$  then the expected size of the matching found by the randomized greedy algorithm on an $d$-regular graph of girth at least $g$ is at least $\frac{(1 - \epsilon)n}{2}$.
\end{corollary}

\corollaryref{cor:gamarnik1} implies that almost every vertex is matched by the randomized greedy algorithm. However, this guarantee is global in nature. For our purpose (e.g., for the proof of \lemmaref{lem:sparsify} that will follow), it is useful to change them into local guarantees, as in \corollaryref{cor:gamarnik}). The local guarantees follow from the fact that the randomized greedy algorithm can be implemented as a (randomized) LCA, which we refer to as RG-LCA, for {\em randomized greedy} LCA. The algorithm is essentially the maximal matching algorithm of~\cite{YYI12}; similarly to~\sectionref{sec:match1}, if the LCA performs more probes than the allotted number, we return ``no''.  

\begin{proposition}
\label{pro:probes}
Let $\epsilon > 0$, $d \geq \max\{\frac{2}{\epsilon}, 3\}$ and $g \geq 8d$, and let $G = (V,E)$ be a $d$-regular graph with girth $g$.
On a query to edge $e \in E$, the probability (over the random permutation over the edges) that  RG-LCA probes extend to edges at distance greater than $g/2$ from $e$ is at most $\frac{\eps}{2d}$.
\end{proposition}

\begin{proof}
 There are exactly $2(d-1)^{\ell}$ edges at distance (exactly) $\ell<g/2$ from any edge $e$, as the graph is $d$-regular. The probability that any specific edge at distance $\ell$ is probed is at most $\frac{1}{\ell!}$. For an edge at distance more than $g/2$ to be probed, at least one edge at distance exactly $g/2-1$ needs to be probed.
 Setting $\ell=g/2-1$, gives that the probability that one such edge is probed is at most (using Stirling's inequality),
 $$\frac{2d^{(4d-1)}}{(4d-1)!}\leq \frac{3}{4}^{4d-1}<\frac{1}{d^2}\leq\frac{\eps}{2d}.$$
\end{proof}

\begin{corollary}
\label{cor:gamarnik}
Let $\epsilon > 0$, $d \geq \max\{\frac{2}{\epsilon}, 3\}$ and $g \geq 8d$, and let $G = (V,E)$ be a $d$-regular graph with girth $g$. For every vertex $v \in V$, the probability that $v$ is matched by RG-LCA is at least $1 - \epsilon-\frac{\eps}{2d}$. For every edge, the probability that it is in this matching is at least $\frac{1-2\eps}{d}$ and at most $\frac{1}{d}$.
\end{corollary}

\begin{proof}
The distance $g/2$ neighborhood of any two vertices or edges in $G$ is identical, due to the regularity and high girth. Therefore, by symmetry, every vertex and edge are equally likely to be in the matching. By~\propref{pro:probes} the probes do not extend beyond this neighborhood, except with negligible probability. This, combined with~\corollaryref{cor:gamarnik1}, gives the vertex probability. As $G$ is $d$-regular, dividing by $d$ gives the edge probabilities.
\end{proof}

We now consider the number of probes made by RG-LCA. For this purpose we use the following theorem of~\cite{YYI12} (an adaptation of their Theorem 2.1 to matching).\footnote{Their theorem (when viewed in the context of SP LCAs) states that the expected number of strong probes an LCA for maximal independent set needs to make is $1+m/n$, where $m$ is the number of edges and $n$ is the number of vertices.  A maximal matching in $G$ is a maximal independent set on the line graph of $G$, $L(G)$. If $G$ has degree bounded by $d$, $L(G)$ has degree bounded by $2d-2$, we have $m\leq (d-1)n$.}

\begin{thm}\cite{YYI12}
\label{thm:YYI}
Let $G$ be a $d$-regular graph. The expected number of strong probes made by RG-LCA is at most $d$ when querying a random edge, where the expectation is over the choice of edge and choice of random permutation.
\end{thm}

When the girth is large, we can remove the dependence on choice of edge in the expectation in ~\theoremref{thm:YYI}:  

\begin{corollary}
\label{cor:YYIa}
Let $G=(V,E)$ be a $d$-regular graph with girth $g \geq 8d$.
For every edge, the probability that RG-LCA uses more than $d^{7/3}$ strong probes when queried on any edge $e \in E$ is at most $d^{-4/3}$, where the expectation is over the choice of random permutation.
\end{corollary}

\begin{proof}
	Using the same symmetry argument as  in the proof of~\corollaryref{cor:gamarnik} gives that for every queried edge, the expected number of strong probes made by RG-LCA is at most roughly $d$. Markov's inequality completes the proof.
	\end{proof}

Suppose now that graph $G$ is $d$-regular,  
but $d > g/8$. In this case~\theoremref{thm:gamarnik} does not guarantee a large matching. 
Likewise, the expected number of probes (as stated in~\theoremref{thm:YYI}) might be too high, if $d$ is not bounded by a polynomial in $\frac{1}{\epsilon}$. We handle both of these issues by the following sparsification procedure that creates an edge induced subgraph $G'$. Let $t = g/10$, and let $d' = t + t^{2/3}$. (We will later also consider the case that $t<g/10$; we will need it when $g$ and $d$ are large. Note that~\propref{pro:deficit} and~\lemmaref{lem:sparsify} hold in this case as well.)

\begin{enumerate}

\item Keep each edge independently with probability $ t/d$.

\item Simultaneously, for every vertex of degree larger than $d'$, remove all of its edges.

\end{enumerate}

In $G'$, for a  vertex $v$ of degree $deg_v$, we define its {\em deficit} to be $d' - deg_v$

\begin{proposition}
\label{pro:deficit}
$G'$ has maximum degree $d'$. The expected sum of deficits for vertices of $G'$ is at most $nt^{2/3} < n(d')^{2/3}$.
\end{proposition}

\begin{proof}
The maximality of the degree is due to the second stage of the sparsification. The expected degree of vertex in $G$ after the first stage is $t$, hence the expected deficit of a vertex in $G'$ is at most $t^{2/3} < d^{2/3}$.
\end{proof}

We want to execute (a modified version of) RG-LCA on $G'$. To be able to apply \corollaryref{cor:gamarnik} we need $G'$ to be regular. For this we imagine that $G'$ is embedded as part of a larger $d'$-regular graph $G''$. The graph $G''$ need not be constructed explicitly. Rather, for the purpose of RG-LCA, each deficit edge of $G'$ is imagined to lead into a fresh $d'$-regular tree of depth much larger than $g$, and that the leaves of these trees are pairwise far away vertices in some huge large girth graph in which these leaves have degree $d' - 1$ and the remaining vertices have degree $d'$. (Queries are only to edges of $G'$. However, such queries might lead to probes to parts of $G'' \setminus G'$. These are ``virtual probes". They are not actual probes to the input graph $G$, but rather to the imaginary graph $G''$ maintained by the LCA.)

Furthermore, we modify RG-LCA by enforcing the following: if a queried edge uses more than $(d')^{7/3}$ probes, the LCA terminates, and the edge is declared not to be in the matching. We refer to this LCA as \emph{bounded-RG-LCA}.

\begin{lemma}
\label{lem:sparsify}
After sparsification, bounded-RG-LCA uses $(d')^{7/3}$ probes and produces a matching of expected size at least $(\frac{1}{2} - O((d')^{-1/3}))n$.
\end{lemma}

\begin{proof}
 The bound on the number of probes is by design. It remains to lower bound the expected size of the matching. In the unbounded-RG-LCA, \corollaryref{cor:gamarnik} implies that every vertex 
has probability at least $(1 - O(1/d'))n$ of being matched. In bounded-RG-LCA,  each edge might be declared unmatched with probability at most $(d')^{-4/3}$, by \corollaryref{cor:YYIa}. We want to bound the expected number of edges that would be matched by RG-LCA but dropped by bounded-RG-LCA.  There  are two reasons this could happen: (1) bounded-RG-LCA exceeds the probe bound or (2) bounded-RG-LCA matches imaginary edges.
We bound each of these separately:
\begin{enumerate}[(1)]
	\item  By \corollaryref{cor:gamarnik}, for any edge $e$, the probability that $e$ is matched by RG-LCA is approximately $1/d'$. Conditioning on the event that an edge $e$ is matched, the probability of $e$ not being added by bounded-RG-LCA is  $O\left( (d')^{-1/3}\right) $. Hence the expected number of matched vertices that remain in $G''$ is  $n(1 - O(d')^{-1/3})$.
	
	\item From~\propref{pro:deficit}, the expected number of deficit edges in $G''$ is at most $n(d')^{2/3}$, hence by~\corollaryref{cor:gamarnik}, the expected number of defective edges  match is at most  $\frac{n(d')^{2/3}}{d'} = n(d')^{-1/3}$.
	\end{enumerate} Overall, the total number of edges that would have been added by RG-LCA but are dropped by bounded-RG-LCA due to is $O((d')^{-1/3})n$.
\end{proof}

We are now ready to prove ~\theoremref{thm:girth}.

\begin{proof}[Proof of~\theoremref{thm:girth}]

 Recall that $d > \frac{1}{\epsilon}$. We consider two possible setting of the parameters, and for each of them bound the number of probes.

\begin{enumerate}

\item $d \le \frac{1}{\epsilon^3}$ and $g > 8d$ (here we do not require $g \ge \frac{1}{\epsilon^3}$). Run bounded-RG-LCA directly on $G$, with a bound of $\frac{d^2}{\epsilon}$ (i.e., if bounded-RG-LCA performs more than $\frac{d^2}{\eps}$ probes, return that the queried edge is unmatched). If we had run (unbounded) RG-LCA, the resulting matching $M$ would have  size $\left( \frac{1}{2} - \frac{\epsilon}{2}\right)n$, from~\corollaryref{cor:gamarnik1}. From~\corollaryref{cor:YYIa}, the number of edges  from $M$ that would not be added by bounded-RG-LCA due to the bound is at most $d^{-2/3}<\eps/2$ in expectation. Hence the expected size of the matching is at least  $ \left( \frac{1}{2}-\eps\right) n$.

\item $d > \frac{1}{\epsilon^3}$ (and $g \ge \frac{1}{\epsilon^3}$). Sparsify $G$ with $t=\Theta\left( \frac{1}{\eps^3}\right)$, satisfying $t \leq \frac{g}{10}$, to obtain a new graph $G'$. Execute bounded- RG-LCA on $G'$ (or rather $G''$ as explained above), enforcing an upper bound of $(d')^{7/3}$ on the number of probes (to $G'$), where $d'=t+t^{2/3}$. By \lemmaref{lem:sparsify}, the size of the matching found is at least $ \left( \frac{1}{2}-\eps\right) n$.
    
    One issue that still needs to be addressed is the number of probes to $G$. By Lemma~\ref{lem:simulation}, every strong probe to $G'$ can be implemented using $d' + 1$ strong probes to $G$. However, for the LCA described above for $G'$, the $d' + 1$ multiplicative overhead in the number of probes can be avoided. The original reason for the $d'+1$ multiplicative overhead was to check whether after the first step of the sparsification procedure, any of the (at most) $d'$ neighboring vertices of a probed vertex has degree above $d'$ (such vertices affect the reply to the strong probe). But in our LCA this check can be omitted. Observe that if at any point our LCA visits a vertex whose degree (after the first sparsification step) is above $d'$, this effectively means that the edge that led the LCA into this vertex should have been discarded (at the second sparsification step) and replaced by an imaginary edge in $G''$. Hence the LCA can simply replace the current vertex by an imaginary vertex and continue without interruption.   (Further details omitted.)
\end{enumerate}
\end{proof}



\paragraph{Acknowledgments.}

We  thank Noga Alon and Mika G\"o\"os for helpful discussions, and the anonymous referees for their useful comments. 
\bibliographystyle{plain}\bibliography{Vardi_PhD_Bibliography}

\begin{thebibliography}{10}

\bibitem{nogapriv}
Noga Alon.
\newblock Private communication, 2016.

\bibitem{ARVX11}
Noga Alon, Ronitt Rubinfeld, Shai Vardi, and Ning Xie.
\newblock Space-efficient local computation algorithms.
\newblock In {\em Proc.\ 22nd {ACM}-{SIAM} {S}ymposium on {D}iscrete
  {A}lgorithms (SODA)}, pages 1132--1139, 2012.

\bibitem{AS08}
Noga Alon and Joel Spencer.
\newblock {\em The Probabilistic Method}.
\newblock John Wiley, 3rd edition, 2008.

\bibitem{Andersen10}
Reid Andersen.
\newblock A local algorithm for finding dense subgraphs.
\newblock {\em {ACM} Trans. Algorithms}, 6(4), 2010.

\bibitem{AGPT16}
Reid Andersen, Shayan~Oveis Gharan, Yuval Peres, and Luca Trevisan.
\newblock Almost optimal local graph clustering using evolving sets.
\newblock {\em J. {ACM}}, 63(2):15, 2016.

\bibitem{BE13}
Leonid Barenboim and Michael Elkin.
\newblock {\em Distributed Graph Coloring: Fundamentals and Recent
  Developments}.
\newblock Synthesis Lectures on Distributed Computing Theory. Morgan {\&}
  Claypool Publishers, 2013.

\bibitem{Blondel05}
Vincent~D Blondel, Julien~M Hendrickx, Alex Olshevsky, and John~N Tsitsiklis.
\newblock Convergence in multiagent coordination, consensus, and flocking.
\newblock In {\em Proceedings of IEEE Conference on Decision and Control},
  pages 2996--3000. IEEE, 2005.

\bibitem{Bollobas65}
Bela Bollobas.
\newblock On generalized graphs.
\newblock {\em Acta Math. Acad. Sci. Hungar.}, 16:447--452, 1965.

\bibitem{CW79}
J.Lawrence Carter and Mark~N. Wegman.
\newblock Universal classes of hash functions.
\newblock {\em Journal of Computer and System Sciences}, 18(2):143 -- 154,
  1979.

\bibitem{CRT05}
Bernard Chazelle, Ronitt Rubinfeld, and Luca Trevisan.
\newblock Approximating the minimum spanning tree weight in sublinear time.
\newblock {\em {SIAM} J. Comput.}, 34(6):1370--1379, 2005.

\bibitem{CV86}
Richard Cole and Uzi Vishkin.
\newblock Deterministic coin tossing with applications to optimal parallel list
  ranking.
\newblock {\em Information and Control}, 70(1):32 -- 53, 1986.

\bibitem{EMR14}
Guy Even, Moti Medina, and Dana Ron.
\newblock Best of two local models: Local centralized and local distributed
  algorithms.
\newblock {\em CoRR}, abs/1402.3796, 2014.

\bibitem{FMS15}
Uriel Feige, Yishay Mansour, and Robert~E. Schapire.
\newblock Learning and inference in the presence of corrupted inputs.
\newblock In {\em Proceedings of The 28th Conference on Learning Theory,
  {COLT}}, pages 637--657, 2015.

\bibitem{FV12}
Fedor Fomin and Yngve Villanger.
\newblock Treewidth computation and extremal combinatorics.
\newblock {\em Combinatorica}, 32(3):289--308, 2012.

\bibitem{FHK16}
Pierre Fraigniaud, Marc Heinrich, and Adrian Kosowski.
\newblock Local conflict coloring.
\newblock In {\em {IEEE} 57th Annual Symposium on Foundations of Computer
  Science, {FOCS}}, pages 625--634, 2016.

\bibitem{GG10}
David Gamarnik and David~A. Goldberg.
\newblock Randomized greedy algorithms for independent sets and matchings in
  regular graphs: Exact results and finite girth corrections.
\newblock {\em Combinatorics, Probability and Computing}, 19:61--85, 1 2010.

\bibitem{Ghaf16}
Mohsen Ghaffari.
\newblock An improved distributed algorithm for maximal independent set.
\newblock In {\em Proceedings of the Twenty-Seventh Annual {ACM-SIAM} Symposium
  on Discrete Algorithms, {SODA} 2016}, pages 270--277, 2016.

\bibitem{GKK10}
Ashish Goel, Michael Kapralov, and Sanjeev Khanna.
\newblock Perfect matchings via uniform sampling in regular bipartite graphs.
\newblock {\em {ACM} Trans. Algorithms}, 6(2):27:1--27:13, 2010.

\bibitem{GKK13}
Ashish Goel, Michael Kapralov, and Sanjeev Khanna.
\newblock Perfect matchings in o(nlog n) time in regular bipartite graphs.
\newblock {\em {SIAM} J. Comput.}, 42(3):1392--1404, 2013.

\bibitem{GPS88}
Andrew~V. Goldberg, Serge~A. Plotkin, and Gregory~E. Shannon.
\newblock Parallel symmetry-breaking in sparse graphs.
\newblock {\em SIAM J. Discret. Math.}, 1(4):434--446, 1988.

\bibitem{GR02}
Oded Goldreich and Dana Ron.
\newblock Property testing in bounded degree graphs.
\newblock {\em Algorithmica}, 32(2):302--343, 2002.

\bibitem{GHLMS15}
Mika G{\"{o}}{\"{o}}s, Juho Hirvonen, Reut Levi, Moti Medina, and Jukka
  Suomela.
\newblock Non-local probes do not help with many graph problems.
\newblock {\em Distributed Computing - 30th International Symposium, {DISC}},
  pages 201--214, 2016.

\bibitem{HP73}
J.~E. Hopcroft and R.~M. Karp.
\newblock An $n^{5/2}$ algorithm for maximum matchings in bipartite graphs.
\newblock {\em SIAM Journal on Computing}, 2(4):225--231, 1973.

\bibitem{JR11}
M.~Jha and S.~Raskhodnikova.
\newblock Testing and reconstruction of {L}ipschitz functions with applications
  to data privacy.
\newblock In {\em Proc.\ 52nd {A}nnual {IEEE} {S}ymposium on {F}oundations of
  {C}omputer {S}cience (FOCS)}, 2011.

\bibitem{Jukna}
Stasys Jukna.
\newblock {\em Extremal Combinatorics with Applications in Computer Science}.
\newblock Springer, 2001.

\bibitem{KCF11}
U.~Kang, Duen~Horng Chau, and Christos Faloutsos.
\newblock Mining large graphs: Algorithms, inference, and discoveries.
\newblock In {\em Proceedings of the 27th International Conference on Data
  Engineering, {ICDE}}, pages 243--254, 2011.

\bibitem{KT00}
J.~Katz and L.~Trevisan.
\newblock On the efficiency of local decoding procedures for error-correcting
  codes.
\newblock In {\em Proc.\ 32nd {A}nnual {ACM} {S}ymposium on the {T}heory of
  {C}omputing (STOC)}, pages 80--86, 2000.

\bibitem{KRS15}
Michael Krivelevich, Daniel Reichman, and Wojciech Samotij.
\newblock Smoothed analysis on connected graphs.
\newblock {\em SIAM J. Discrete Math}, 29(3), 2015.

\bibitem{KMW04}
Fabian Kuhn, Thomas Moscibroda, and Roger Wattenhofer.
\newblock What cannot be computed locally!
\newblock In {\em Proceedings of the Twenty-Third Annual {ACM} Symposium on
  Principles of Distributed Computing, {PODC}}, pages 300--309, 2004.

\bibitem{KMW16}
Fabian Kuhn, Thomas Moscibroda, and Roger Wattenhofer.
\newblock Local computation: Lower and upper bounds.
\newblock {\em J. {ACM}}, 63(2):17, 2016.

\bibitem{LRR14}
Reut Levi, Dana Ron, and Ronitt Rubinfeld.
\newblock Local algorithms for sparse spanning graphs.
\newblock In {\em Approximation, Randomization, and Combinatorial Optimization.
  Algorithms and Techniques (APPROX/RANDOM)}, pages 826--842, 2014.

\bibitem{LRY15}
Reut Levi, Ronitt Rubinfeld, and Anak Yodpinyanee.
\newblock Brief announcement: Local computation algorithms for graphs of
  non-constant degrees.
\newblock In {\em Proceedings of the 27th {ACM} on Symposium on Parallelism in
  Algorithms and Architectures, {SPAA}}, pages 59--61, 2015.

\bibitem{Linial92}
Nathan Linial.
\newblock Locality in distributed graph algorithms.
\newblock {\em SIAM J. Comput.}, 21(1), 1992.

\bibitem{PCVW17}
Palma London, Niangjun Chen, Shai Vardi, and Adam Wierman.
\newblock Distributed optimization via local computation algorithms.
\newblock \url{http://users.cms.caltech.edu/~plondon/loco.pdf}, 2017.

\bibitem{LPP08}
Zvi Lotker, Boaz Patt-Shamir, and Seth Pettie.
\newblock Improved distributed approximate matching.
\newblock In {\em Proc.\ 20th ACM {S}ymposium on {P}arallel {A}lgorithms and
  {A}rchitectures (SPAA)}, pages 129--136, 2008.

\bibitem{Low2002}
Steven~H Low, Fernando Paganini, and John~C Doyle.
\newblock Internet congestion control.
\newblock {\em IEEE control systems}, 22(1):28--43, 2002.

\bibitem{MPV15}
Yishay Mansour, Boaz Patt{-}Shamir, and Shai Vardi.
\newblock Constant-time local computation algorithms.
\newblock In {\em Approximation and Online Algorithms - 13th International
  Workshop, {WAOA}}, pages 110--121, 2015.

\bibitem{MRVX12}
Yishay Mansour, Aviad Rubinstein, Shai Vardi, and Ning Xie.
\newblock Converting online algorithms to local computation algorithms.
\newblock In {\em Proc.\ 39th {I}nternational {C}olloquium on {A}utomata,
  {L}anguages and {P}rogramming (ICALP)}, pages 653--664, 2012.

\bibitem{MV13}
Yishay Mansour and Shai Vardi.
\newblock A local computation approximation scheme to maximum matching.
\newblock In {\em APPROX-RANDOM}, pages 260--273, 2013.

\bibitem{NN90}
Joseph Naor and Moni Naor.
\newblock Small-bias probability spaces: Efficient constructions and
  applications.
\newblock In {\em STOC}, pages 213--223, 1990.

\bibitem{NS95}
Moni Naor and Larry~J. Stockmeyer.
\newblock What can be computed locally?
\newblock {\em SIAM J. Comput.}, 24(6):1259--1277, 1995.

\bibitem{NO08}
Huy~N. Nguyen and Krzystof Onak.
\newblock Constant-time approximation algorithms via local improvements.
\newblock In {\em Proc.\ 49th {A}nnual {IEEE} {S}ymposium on {F}oundations of
  {C}omputer {S}cience (FOCS)}, pages 327--336, 2008.

\bibitem{PR07}
M.~Parnas and D.~Ron.
\newblock Approximating the minimum vertex cover in sublinear time and a
  connection to distributed algorithms.
\newblock {\em Theoretical Computer Science}, 381(1--3), 2007.

\bibitem{Peleg00}
David Peleg.
\newblock {\em Distributed Computing: A Locality-Sensitive Approach}.
\newblock SIAM Monographs on Discrete Mathematics and Applications, 2000.

\bibitem{RV16}
Omer Reingold and Shai Vardi.
\newblock New techniques and tighter bounds for local computation algorithms.
\newblock {\em J. Comput. Syst. Sci.}, 82(7):1180--1200, 2016.

\bibitem{RTVX11}
Ronitt Rubinfeld, Gil Tamir, Shai Vardi, and Ning Xie.
\newblock Fast local computation algorithms.
\newblock In {\em Proc.\ 2nd {S}ymposium on {I}nnovations in {C}omputer
  {S}cience (ICS)}, pages 223--238, 2011.

\bibitem{SS10}
Michael~E. Saks and C.~Seshadhri.
\newblock Local monotonicity reconstruction.
\newblock {\em {SIAM} J. Comput.}, 39(7):2897--2926, 2010.

\bibitem{ST13}
Daniel~A. Spielman and Shang{-}Hua Teng.
\newblock A local clustering algorithm for massive graphs and its application
  to nearly linear time graph partitioning.
\newblock {\em {SIAM} J. Comput.}, 42(1):1--26, 2013.

\bibitem{Vad12}
Salil~P. Vadhan.
\newblock Pseudorandomness.
\newblock {\em Foundations and Trends in Theoretical Computer Science},
  7(1--3):1--336, 2011.

\bibitem{VardiPhD}
Shai Vardi.
\newblock {\em Designing Local Computation Algorithms and Mechanisms}.
\newblock PhD thesis, Tel Aviv University, Tel Aviv, Israel, 2015.

\bibitem{Vaz01}
Vijay~V. Vazirani.
\newblock {\em Approximation Algorithms}.
\newblock Springer, 2001.

\bibitem{Yao77}
Andrew Chi-Chin Yao.
\newblock Probabilistic computations: Toward a unified measure of complexity.
\newblock In {\em Proceedings of the 18th Annual Symposium on Foundations of
  Computer Science}, FOCS '77, pages 222--227, 1977.

\bibitem{Yekh12}
Sergey Yekhanin.
\newblock Locally decodable codes.
\newblock {\em Foundations and Trends in Theoretical Computer Science},
  6(3):139--255, 2012.

\bibitem{YYI12}
Yuichi Yoshida, Masaki Yamamoto, and Hiro Ito.
\newblock Improved constant-time approximation algorithms for maximum matchings
  and other optimization problems.
\newblock {\em {SIAM} J. Comput.}, 41(4):1074--1093, 2012.

\end{thebibliography}

\pagenumbering{roman}
\section*{APPENDIX}
\appendix

\section{The (Modified) Goldberg-Plotkin-Shannon Algorithm}
\label{app:goldberg}

We present the simplified version of the Golberg-Plotkin-Shannon algorithm \cite{GPS88} for unicyclic trees, which we call MGPS. The two modifications are the removal of the roots as a special case, and the execution of the algorithm for some possibly ``redundant'' rounds; the GPS algorithm terminates as soon as the number of colors is at most $6$, while the MGPS algorithm may continue for several more rounds. This can cause the colors generated by the GPS and MGPS algorithms to be different, but it does not affect the correctness of the algorithm. The algorithm consists of two stages. In the first, shown as \algorithmref{alg:sixcol}, each unicyclic  tree is colored with $6$ colors and in the second, the number of colors is reduced to $3$, by \emph{shifting down}: when there are $x \in \{3,4,5\}$ colors, each vertex takes the color of its parent, and then vertices colored $x$ change their color to the smallest color in $\{0,1,2\}$ that is different than their parents' and children's current color (which is their previous color).

\vspace{10pt}
\begin{algorithm}[H]
	\caption{Parallel $6$ Coloring a Unicyclic-Forest}\label{alg:sixcol}
	\SetKwInOut{Input}{Input}\SetKwInOut{Output}{Output}\SetKwInOut{Inquiry}{Inquiry}
	
	\Input{$G=(V,E)$}
	\Output{a color for each vertex}
	
	\BlankLine
	Let $\T$ be the maximal number of rounds required to reduce the number of colors of an $n$ vertex graph to $6$\;
	\ForEach{$v \in V$ in parallel}{$C_v = ID(v)$\;}
	$i=0$\;
		\While{$i\leq\T$}
		{
			\ForEach{$v \in V$ in parallel}
			{	
			
					$a_v = \min\{j|C_v(j) \neq C_{parent(v)}(j)\}$\;
					$b_v=C_v(a_v)$\;
				$C_v = (a_v, b_v)$\;
			}
			
			$i=i+1$\;
		}

\end{algorithm}

\vspace{10pt}
The following theorem asserts the correctness and running time of 	\algorithmref{alg:sixcol}.
\begin{theorem}
	\algorithmref{alg:sixcol} followed by the shifting-down phase produces a valid $3$-coloring of a unicyclic  tree in $O(\log^*{n})$ time on a CREW PRAM using a linear number of processors.
\end{theorem}


\begin{proof}[Proof sketch]
    We prove by induction that every vertex is colored differently than its parent. For the base of the induction, at the start of the algorithm, each vertex is colored with its ID, and the hypothesis holds.
    For the inductive step, assume that vertex $v$ is colored differently than its parent, $p(v)$ at the start of a round. Denote $p(v)$'s parent by $p(p(v))$ (note that possibly $p(p(v))=v$). If the first bit on which the colors of $v$ and $p(v)$ differ is not the same as the first bit on which the colors of $p(v)$ and $p(p(v))$ differ, $i_v\neq i_{p(v)}$, and we are done. Otherwise, clearly $b_v \neq b_{p(v)}$. 
      Similarly, shifting down preserves the coloring invariant. 
      
      Regarding the complexity: in each round of 	\algorithmref{alg:sixcol}, the  number of bits needed to represent a vertex's color decreases logarithmically, implying at most $O(\log^*{n})$ rounds are needed. Shifting down requires $3$ more rounds. 
	\end{proof}

\section{On the Number of Connected Subgraphs}

\label{sec:conn}

Given a graph $G(V,E)$ on $n$ vertices, a set $S \subset V$ of vertices is referred to as {\em connected} if its induced subgraph is connected. The neighborhood $N(S)$ contains those vertices in $V \setminus S$ that have at least one neighbor in $S$.
Let $C(s,t,n)$ denote the maximum number of connected subsets $S$ of size $|S| = s$ that have precisely $t$ neighbors (namely, $|N(S)| = t$) in an $n$ vertex graph.

\begin{thm}
\label{thm:graphs}
Using the notation above, $C(s,t,n) \le \frac{n}{s}{s+t-1 \choose t}$.
\end{thm}

Proofs for Theorem~\ref{thm:graphs} (sometimes with the term $\frac{n}{s}$ replaced by $n$) can be found in several places (e.g.,~\cite{FV12} and~\cite{KRS15}), but we do not know what the earliest reference to this theorem is. Alon~\cite{nogapriv} points out that Theorem~\ref{thm:graphs} is basically a special case of the following theorem of Bollobas~\cite{Bollobas65}.

\begin{thm}
\label{thm:sets}
Let $U$ be a collection of items and let $A_1, \ldots, A_m, B_1, \ldots B_m$ be a collection of sets such that:

\begin{enumerate}

\item $|A_i| = a$ for every $1 \le i \le m$.

\item $|B_i| = b$ for every $1 \le i \le m$.

\item $A_i \cap B_i = \emptyset$ for every $1 \le i \le m$.

\item $A_i \cap B_j \not= \emptyset$ for every $i \not= j$.

\end{enumerate}

Then $m \le {a + b \choose b}$ holds. In particular, the upper bound on the number of sets is independent of the number $|U|$ of items.
\end{thm}

For completeness, let us present a known short proof (that appears in Section~1.3 of~\cite{AS08} and in Section 9.2.2 of~\cite{Jukna}) of Theorem~\ref{thm:sets}.

\begin{proof}
Consider a uniformly random permutation of $U$. For $1 \le i \le m$, let $x_i$ be an indicator random variable for the event that all items of $A_i$ precede all items of $B_i$ in this random permutation, and let $X = \sum_{i=1}^m x_i$. Then $Pr[x_i = 1] = \frac{a! b!}{(a+b)!} = \frac{1}{{a + b \choose b}}$ and consequently $E[X] =  \frac{m}{{a + b \choose b}}$. However,  the combination of the third and fourth conditions from Theorem~\ref{thm:sets} imply that the events underlying $x_i$ and $x_j$ are disjoint for every $i \not= j$, and consequently $X \le 1$. Hence $m \le {a + b \choose b}$, as desired.
\end{proof}

We now prove Theorem~\ref{thm:graphs} using Theorem~\ref{thm:sets}.

\begin{proof}
Fix an arbitrary vertex $v$. A set $S_i \subset V$ will be referred to as {\em eligible} for $v$ if $S_i$ is connected, $v\in S_i$, and the neighborhood $N(S_i)$ has size exactly $t$. Let $q_v$ be the number of sets that are eligible for $v$. Letting the sets $S_i \setminus v$ serve as $A_i$ in Theorem~\ref{thm:sets} and the sets $N(S_i)$ serve as $B_i$, it is not hard to see that all conditions of Theorem~\ref{thm:sets} are satisfied, with $a = s-1$ and $b= t$. It follows that $q_v \le {s+t-1 \choose t}$.

To complete the proof, sum $q_v$ over all $n$ vertices of $G$, and note that every eligible set $S_i$ is counted exactly $s$ times.
\end{proof}

\section{Random Weak $2$-Coloring Seed Length}

\label{app:kwise}

In order to prove~\theoremref{thm:ak}, we need several definitions.
A pseudorandom generator\footnote{We refer the reader to ~\cite{Vad12} for an excellent tutorial on pseudorandomness.} is an algorithm that takes as input a 
short, perfectly random seed and returns a (much longer) sequence of bits that ``looks'' random.  We clearly sacrifice some randomness when we do this: the bit sequence we get is defined by the random seed, and if the seed is of length $s$, there are only $2^s$ possible distinct values of the sequence.
The notions of $k$-wise independent hash functions and almost $k$-wise independent hash functions were introduced by  Carter and Wegman \cite{CW79}. 
To quantify what we mean by almost $k$-wise independence, we use the notion of \emph{statistical distance}.
\begin{definition}[Statistical distance]\label{def:stat}
	For random variables $X$ and $Y$ taking values in $\mathcal{U}$, their \emph{statistical distance}  is $$\Delta(X, Y) = max_{D\subseteq \mathcal{U}}|\Pr[X \in D] - \Pr[Y \in D]|.$$ For $\eps\geq 0$, we say that $X$ and
	$Y$ are $\eps$-close if $\Delta(X, Y)\leq \eps$.
\end{definition}

\begin{definition}[$\eps$-almost $k$-wise independent hash functions] For $n,L,k \in \N$ such that $k \leq n$, let $Y$ be a random variable sampled uniformly at random from $[L]^k$. For $\eps\geq 0$, a family of functions $\mcH = \{h : [n] \rightarrow [L]\}$ is \emph{$\eps$-almost $k$-wise independent} if for all distinct $x_1, x_2, \ldots ,x_k \in [n]$, we have that $\langle h(x_1), h(x_2),\ldots ,h(x_k)\rangle $ and $Y$ are $\eps$-close,  when
	$h$ is sampled uniformly from $\mcH$.  \end{definition}

The following lemma is the main building block for the proof  of~\theoremref{thm:ak}. 
\begin{lemma}\label{lemma:gg}
	Let $k=4\log n$ 
	and let $\eps = 1/{n^4}$. Let $c_{temp}$ be a function sampled uniformly at random from a family of $\eps$-almost $k$-wise independent hash functions. Then~\algorithmref{alg:lca2rand}  makes at most $k$  probes with probability at least $1-1/{n^2}$.
\end{lemma}

\begin{proof}
	Fix $G=(V,E)$ and the queried vertex $v_0\in V$.  Note that in any implementation of~\algorithmref{alg:lca2rand}, the algorithm stops making probes to $G$ when either (1) a primary root is found: a vertex ID $u$ is observed such that $c_{temp}(u) \neq c_{temp}(v_0)$ or (2) a secondary root is found.  
	Denote by $\A_k$ the algorithm that executes~\algorithmref{alg:lca2rand}, except that if it has observed the IDs of fewer than $k$ vertices before finding a root, it continues probing the graph until it has seen exactly $k$ IDs: If it finds a primary  root, it continues making exactly the same probes as~\algorithmref{alg:lca2rand} would have made if for every observed vertex $u$, $c_{temp}(u)=c_{temp}(v_0)$. If it encounters a secondary root, it continues making probes according to some (fixed) predetermined rule (e.g., probe all the neighbors of the vertex with the smallest ID out of all observed vertices, then the neighbors of the vertex with the second smallest ID, and so on).
	
	Let $v_1,\ldots,v_k$ be the first $k$ vertex IDs (other than $v_0$) that $\A_k$ observes. 
	We want to bound the probability that $c_{temp}(v_0)= c_{temp}(v_1) = \cdots = c_{temp}(v_k)$, an event we denote by $X$. This is an upper bound on the probability that $\A_k$ needs to make more than $k$ probes, hence also an upper bound on the probability that~\algorithmref{alg:lca2rand} needs to make more than $k$ probes. Let $g$ be truly random independent assignment of $\{0,1\}$ to $V$. The probability that $g(v_i)=g(v_j)$ for all $i,j \in \{0,1, \ldots, k\}$ is $2^{-k}$. As $c_{temp}$ is an $\eps$-almost $k$-wise independent hash function, $\Pr[X]  \leq 2^{-k}+\eps$; therefore $\Pr[X] <1/n^3$. Taking a union bound over all possible queries gives that the probability that there exists a query for which~\algorithmref{alg:lca2rand} needs to make more probes after it has observed $k$ vertex IDs is at most $1/n^2$. 
\end{proof}

We require the following theorem of Naor and Naor~\cite{NN90}.
\begin{theorem} [\cite{NN90}]\label{thm:almost}
	For every $n,k \in \N$ and $\eps > 0$ such that $n$ is a power of $2$,  and $\eps>0$, there is a family of  $\eps$-almost $k$-wise  independent functions
	$\mathcal{H} =  \{h : [n] \rightarrow \{0,1\}\}$  whose seed length is  $O(\log\log{n}+k+
	\log(1/\eps))$. In particular, for  $k=O(\log n)$ the seed length is $O(\log{n}+\log(1/\eps))$. 
\end{theorem}

\begin{proof}[Proof of~\theoremref{thm:ak}]
	Let $c_{temp}$ be a function sampled uniformly at random from a family of hash functions as in~\theoremref{thm:almost}, setting $\eps=1/n^4$ and $k=4\log{n}$. The seed length for such a family is $O(\log{n})$. 
	From~\lemmaref{lemma:gg},~\algorithmref{alg:lca2rand}  observes at most $k$ vertices with probability $1-1/n^2$.
	
	In order to observe $c\log{n}$ IDs,~\algorithmref{alg:lca2rand}, with arbitrary parent choice, needs to make at most $c\log{n}$ strong probes or $c(c+1)\log{n}^2$ weak probes. Therefore with probability at least $1-1/n^2$,~\algorithmref{alg:lca2rand} will make at most $4\log{n}$ strong probes or $20\log^2{n}$ weak probes. 
\end{proof}

\section{Almost-Maximum Matching on Bounded Degree Graphs}
\label{app:bounded}

Nguyen and Onak~\cite{NO08} give an algorithm that approximates the size of the maximum matching in graphs with degree bounded by $d \geq 2$ to within $(1-\eps)$ using $2^{d^{O\left( 1/\eps\right)} }$ probes. The algorithm samples $O(1/\eps^2)$ edges uniformly at random, and checks whether they are in some matching $M^*$, using the algorithm  detailed below (the `O' notation hides a dependency on $d$ that is not analyzed). By the Hoeffding bound, this suffices to estimate $|M^*|$ with constant probability and additive error at most $\eps n$. Yoshida, Yamamoto and Ito~\cite{YYI12} show how to improve the total number of probes to $d^{O\left(\frac{1}{\eps^2} \right) }\left( \frac{1}{\eps}\right) ^{O\left( \frac{1}{\eps}\right)}$, by tweaking the algorithm of~\cite{NO08}, and fixing the number of vertices queried to be $O(d^2/\eps^2)$. The algorithms build on the algorithm of Hopcroft and Karp~\cite{HP73}, and we summarize them here for completeness. Note that others (e.g.,~\cite{LPP08,MV13,LRY15}) have used similar algorithms in the context of approximation algorithms, distributed algorithms and LCAs. We first need some definitions.

An {\em augmenting path} with respect to a matching $M$ is a simple path
whose endpoints are \emph{free} (i.e., not part of any edge in the
matching $M$), and whose edges alternate between $E\setminus M$ and
$M$.
A set of augmenting paths $P$ is  {\em independent} if no two paths
$p_1,p_2\in P$ share a vertex.

For sets $A$ and $B$, we denote $A \oplus B \overset{def}{=} (A \cup
B) \setminus (A \cap B)$. An important observation regarding
augmenting paths and matchings is the following.

\begin{obs}
	If $M$ is a matching and $P$ is an independent set of augmenting
	paths, then $M\oplus P$ is a matching of size $|M|+|P|$.
\end{obs}
While the main result of Hopcroft and Karp \cite{HP73} is an
improved matching algorithm for bipartite graphs, they show the
following useful lemmas for general graphs.

\begin{lemma}\cite{HP73} \label{lemma:hp1}
	Let $G=(V,E)$ be an undirected graph, and let $M$ be some matching in $G$.
	If the shortest augmenting path with respect to $M$ has length $\ell$ and
	$\Phi$ is a maximal set of independent augmenting paths of length
	$\ell$, the shortest augmenting path with respect to $M \oplus \Phi$ has
	length strictly greater than $\ell$.
\end{lemma}

\begin{lemma}\cite{HP73} \label{lemma:hp2}
	Let $G=(V,E)$ be an undirected graph. Let $M$ be some matching in
	$G$, and let $M^*$ be a maximum matching in $G$. If the shortest
	augmenting path with respect to $M$ has length $2k - 1 > 1$ then
	$|M| \geq (1-1/k)|M^*|$.
\end{lemma}

The high level of the approach is the folllwing: Start with an empty matching.
In stage $\ell = 1,3,\ldots, 2k-1$, add a maximal independent
collection of augmenting paths of length $\ell$, to obtain the matching $M_\ell$. For $k=\lceil
1/\eps \rceil$, by Lemma \ref{lemma:hp2}, we have that the matching
$M_{2k-1}$ is a $(1-\epsilon)$-approximation to the maximum
matching.

In order to find such a collection of augmenting paths of
length $\ell$, we need to define a conflict graph:
\begin{definition}\cite{LPP08} Let $G = (V,E)$ be an undirected graph, let $M \subseteq E$
	be a matching, and let $\ell > 0$ be an integer.  The
	$\ell$-conflict graph with respect to $M$ in $G$, denoted
	$C_M(\ell)$, is defined as follows. The nodes of $C_M(\ell)$ are all
	augmenting paths of length  $\ell$, with respect to $M$, and
	two nodes in $C_M(\ell)$ are connected by an edge if and only if
	their corresponding augmenting paths intersect at a vertex of
	$G$.
	\end{definition}

The high level sequential (non-local) algorithm for computing a maximal matching used by \cite{NO08, YYI12} (and others) is described as~\algorithmref{alg1}.


\begin{algorithm}[h]
	\caption{ \textsc{High Level Maximum Matching Approximation Algorithm}}\label{alg1}
	\SetKwInOut{Input}{Input}\SetKwInOut{Output}{Output}
	\LinesNumbered
	\Input{$G=(V,E)$ and $\eps>0$}
	\Output{A matching $M$}
	\BlankLine
	
	$M_{-1}\gets \emptyset$ \tcp*{$M_{-1}$ is the empty matching}
	$k\gets \lceil 1/\eps \rceil$~\;
	\For{$\ell= 1, 3, \ldots, 2k-1$}{
		Construct the conflict graph $C_{M_{\ell-2}}(\ell)$~\;\label{line:constr}
		Let $\mathcal{I}$ be an MIS of $C_{M_{\ell-2}}(\ell)$~\; \label{line:MIS}
		Let $\Phi(M_{\ell-2})$ be the union of augmenting paths corresponding to $\mathcal{I}$~\; \label{line:P}
		$M_{\ell}\gets M_{\ell-2} \oplus \Phi(M_{\ell-2})$ \tcp*{$M_{\ell}$ is matching at the end of phase $\ell$} \label{line:last}
	}
	Output $M_{\ell}$\tcp*{$M_{\ell}$ is a $(1-\frac{1}{k+1})$-approximate maximum matching}
\end{algorithm}

The idea behind~\cite{NO08,YYI12} is to implement oracle access to the matchings $M_\ell$ in~\algorithmref{alg1}. Recall that the goal is to output an approximate size of the maximum matching. To that end,  query $O(d^2/\eps^2)$ edges at random. For each one,  find out if it is in $M_{2k-1}$. The proportion of edges that \emph{are} in $M_{2k-1}$ is an approximation to the size of the maximum matching. To find out whether an edge is in $M_\ell$, $\ell=1, 3, \ldots, 2k-1$, construct $M_{\ell-2}$  (note that $\ell$ is always odd) and use that to construct the conflict graph: Let $H_{\ell-2}$ be the graph constructed by associating every path $p_i$ of length $\ell-2$ with a vertex, and adding an edge between two vertices $p_i$ and $p_j$ only if there is some vertex $v \in G$ that is in both $p_i$ and $p_j$. We say that a vertex $p_i$ in $H_{\ell-2}$ is \emph{valid} if $p_i$ is an augmenting path with respect to $M_{\ell-2}$. The conflict graph $C_{M_{\ell-2}}$ is the subgraph of $H_{\ell-2}$ induced by the valid vertices of $H_{\ell-2}$. 
In order to find all paths of length $\ell$ containing a vertex $v$, it suffices to probe the graph $d^{2\ell}$ times. The difficulty comes from finding a maximal independent set in each of these conflict graphs. Nguyen and Onak do the following: Let $G=(V,E)$ be a graph. For each vertex $v \in V$, generate a unique real number in $[0,1]$, called its \emph{rank}. Construct a  \emph{query subgraph} rooted at vertex $v$, denoted $T_v$, is constructed as follows.  Initialize $T_v$ to contain $v$, and  add to $T_v$ all of the neighbors of $v$ whose rank is smaller than the rank of $v$. Continue iteratively - for each vertex $u$ in $T_v$,  add all of its neighbors whose rank is at most the rank of $u$: $\{w: (w,u) \in E, r(w)\geq r(u)\}$. Then, simulate the greedy online algorithm for MIS on this subgraph, such that lower ranked vertices arrive earlier. The size of the query subgraph\footnote{More precisely, the size of the query subgraph multiplied by $d$, as it is necessary to probe the neighbors of the perimeter of the subgraph in order to determine its borders.} rooted at $v$ is an upper bound on the number of probes made to the graph in order to compute whether $v$ is in the MIS. Nguyen and Onak use a ``locality lemma'' in order to bound the expected  number of probes that need to be made.
Markov's inequality is used to show that  if $O(1/\eps^2)$ queries are made, the total number of probes will be at most $2^{d^{O(1/\eps)}}$ with probability at least $5/6$. If more probes are needed, the algorithm can be terminated and started again with new random bits. 

\cite{YYI12} improve the subroutine for maximal independent set as follows: instead of computing the entire query subgraph up front, only as much of the subgraph as is needed is constructed. For example, assume that the queried vertex $u$ has rank $0.8$, and its neighbors $v_1, v_2$ and $v_3$ have rank $0.1, 0.4$ and $0.7$ respectively. $v_1$ is added to the query subgraph first, and  checked recursively for membership in the independent set. If it is in the MIS, we can return ``no'', as $u$ is certainly not in the MIS. If $v_1$ is not in the MIS, $v_2$ is added and checked, and so on. Clearly in many cases fewer probes are made than if the entire query subgraph was discovered in advance, but there is another advantage: lower ranked vertices are more likely to be in the independent set, and these are probed first. The algorithm uses $d^{6k^2} k^{O(k)}$ calls to the MIS subroutine, which uses $O(d^2)$ weak probes. We note that~\cite{YYI12} use the weak probe model, and their MIS algorithm can be implemented using $O(d)$ strong probes in expectation. Overall,

\begin{lemma}[\cite{YYI12}]
	Let $G$ be a graph whose degree is bounded by $d$. The expected number of probes used to reply to a query $e \in  M_{2k-1}$ is $d^{6k^2} k^{O(k)}$.
	\end{lemma}

In~\cite{YYI12}, similarly to~\cite{NO08},  if the algorithm (for approximating the size of the maximum matching) uses more probes than the allotted number, it can be terminated and started again with new randomness.  With LCAs that is not the case. We overcome this issue by allowing the LCA to use $c$ times more probes, for some constant $c>0$. If the process still hasn't terminated, we return ``no'' as a reply to the query. 

\section{Expectation of the Maximum of Two Random Variables}

\begin{claim}\label{claim:expectmax}
	Let $X$ and $Y$ be two (possibly correlated) non-negative random variables drawn from the same distribution, with mean $\mu$ and variance $\sigma^2$. Then
$\expect\left[ \max\{X, Y\}\right] \leq \mu + 3\sigma$.
\end{claim}
\begin{proof}

	By Chebyshev's inequality, it holds that 
	$\Pr[|X-\mu| >a]\leq \frac{\sigma^2}{a^2}$. 
	Set $Z=\max\{X,Y\}$.
	
	By the union bound, it holds that $\Pr[Z-\mu>a] \leq \frac{2\sigma^2}{a^2}$. 

	The expectation of $Z$ is
	\begin{align*}
	\expect[Z]
	&= \int_{0}^\infty  \Pr[Z>a] da\\
	&\leq \mu + \sigma + \int_{\sigma}^\infty Pr[Z-\mu >a] da\\
	&\leq \mu + \sigma + 2\int_{\sigma}^\infty  \frac{\sigma^2}{a^2} da\\
	&\leq \mu + 3\sigma.
	\end{align*}
	\end{proof}

\section{Almost Maximum Matching in Almost $d$-Regular Graphs}

\label{sec:mmm}

Let $G(V,E)$ be a graph with $n$ vertices and $m$ edges. Let $d = \frac{2m}{n}$ denote its average degree, and let $0 \le \delta < \frac{1}{2}$ be such that the degree of every vertex is in the range $[(1-\delta)d, (1+\delta)d]$. For $\delta = 0$ the graph $G$ is regular, and otherwise $G$ is approximately regular.

The following procedure will be referred to as {\em randomized bipartite sparsification} (with parameter $q$). It constructs from $G(V,E)$ a random bipartite graph $G'(V_0,V_1;E')$ in three steps:

\begin{enumerate}

\item The set $V$ of vertices is partitioned at random into two groups $V_0$ and $V_1$. This is done by choosing a random function $R_1: V \longrightarrow \{0,1\}$, where for each vertex $v$ independently it holds that $R_1(v) = 0$ with probability $\frac{1}{2}$. Vertex $v$ belongs to $V_{R_1(v)}$. Edges in $E$ with one endpoint in each side are referred to as {\em bipartite} edges and remain in the graph. The other edges are referred to as {\em monochromatic} edges and are dropped from the graph.

\item Every bipartite edge is declared to be a {\em surviving} edge independently with probability $\frac{q}{d}$. The other edges are dropped from the graph. Determining which edges are surviving is done by choosing a random function $R_2: E \longrightarrow \{0,1\}$, where for each edge $e$ independently it holds that $R_2(e) = 1$ with probability $\frac{q}{d}$. Edge $e$ survives if $R_2(e) = 1$.

\item To obtain $E'$, (simultaneously) drop all those surviving edges that are incident with a vertex that has more than $q$ surviving edges.

\end{enumerate}

\begin{proposition}
$G'$ is bipartite and has maximum degree $q$.
\end{proposition}

\begin{proof}
Bipartiteness follows from Step~1 of the randomized bipartite sparsification. The bound on the maximum degree follows from Step~3, in which vertices with more than $q$ surviving edges lose all their edges.
\end{proof}

\begin{lemma}
\label{lem:simulation1}
Given $R_1$ and $R_2$ of the randomized bipartite sparsification procedure, every strong probe in $G'$ can be implemented by $q+1$ strong probes to $G$.
\end{lemma}

\begin{proof}
On a probe to vertex $v \in V$ (where $V = V_1 \cup V_2$), the reply needs to be the list of its neighbors in $G'$. Probe $v$ in $G$ so as to get the list of its neighbors in $G$, and hence also the list of its incident edges in $G$. Apply $R_1$ to $v$ and each of its neighbors, and determine the bipartite edges. Apply $R_2$ to the bipartite edges, and determine which of them is surviving. If $v$ has more than $q$ surviving edges, then reply that $v$ has no neighbors in $G'$. If $v$ has at most $q$ surviving edges, probe each of their endpoints, determine as above how many surviving edges it has, and include the respective vertex in the reply only if it has at most $q$ surviving edges.
\end{proof}

\begin{thm}
\label{thm:sparsify}
There are universal constants $c_1, c_2, c_3 \ge 0$ such that for every $0 < \epsilon \le \frac{1}{4}$ the following holds. Let $G$ be a graph on $n$ vertices, with average degree $d  > c_1\frac{\log \frac{1}{\epsilon}}{\epsilon^2}$ and degrees in the range $[(1 - \delta)d, (1 + \delta)d]$ where $\delta \le c_2\epsilon$. Then the expected size of a maximum matching in graph $G'$ (generated by randomized bipartite sparsification with degree bound $q \ge \frac{c_3}{\epsilon}$) is at least $(\frac{1}{2} - \epsilon)n$.
\end{thm}

\begin{proof}
We sketch the proof, omitting some of the computations, and without optimizing the constants involved. Also, we will consider expectations of various quantities, omitting the proofs that their realized value is w.h.p. close to their expectation.

Consider first the outcome of Step~1. Since every edge has probability exactly $\frac{1}{2}$ of being bipartite, the expected average degree of the bipartite graph is $\frac{d}{2}$. Moreover, for some fixed constant $c > 0$ and every constant $k > 2$, for every vertex, its degree will fall in the interval $[(1 - \delta)\frac{d}{2} - c\sqrt{d\log k}, (1 + \delta)\frac{d}{2} + c\sqrt{d\log k}]$ with probability at least $1 - \frac{1}{k}$. Hence:

\begin{enumerate}
\item The expected number of vertices of degree smaller than $(1 - \delta)\frac{d}{2} - c\sqrt{d\log k}$ is at most $\frac{n}{k}$.
\item The expected total number of edges incident with vertices of degree larger than $(1 + \delta)\frac{d}{2} + c\sqrt{d\log k}$ is at most $(1 + \delta)\frac{dn}{k}$.
    \end{enumerate}

Now let us analyze the expected size of the maximum (bipartite) matching in $G'$ after Step~2. Let $s \ge 0$  be such that this maximum matching contains $\frac{n}{2} - s$ edges. Then $2s$ vertices are unmatched. Without loss of generality, at least $s$ of the unmatched vertices are in $V_0$. By Hall's condition, this implies that for some $\ell$, there is a set $S \subset V_0$ of size $\ell + s$, and a set $T$ in $V_1$ of size $\ell$, such that all surviving edges with one endpoint in $S$ have their other endpoint in $T$. The probability that such an event happens can be upper bounded as follows. There are at most $3^n$ ways of choosing the partition of $V$ into $(S,T,V \setminus (S \cup T))$.
After Step~1 (and using item~1 above), the sum of degrees of vertices in $S$ is at least $\left((1 - \delta)\frac{d}{2} - c\sqrt{d\log k}\right)\left(\ell + s - \frac{n}{k}\right)$. Using item~2 above, the sum of degrees of vertices in $T$ is at most $\left((1 + \delta)\frac{d}{2} + c\sqrt{d\log k}\right)\ell + (1 + \delta)\frac{dn}{k}$. Hence after Step~1, the number of edges with one endpoint in $S$ and no endpoint in $T$ is at least:

\begin{eqnarray*}
  \left((1 - \delta)\frac{d}{2} - c\sqrt{d\log k}\right)\left(\ell + s - \frac{n}{k}\right) - \left((1 + \delta)\frac{d}{2} + c\sqrt{d\log k}\right)\ell - (1 + \delta)\frac{dn}{k} \\
  \ge \left((1 - \delta)\frac{d}{2} - c\sqrt{d\log k}\right)s - \left( \delta d + 2c\sqrt{d\log k}\right) \ell - \frac{2dn}{k} \\
  \ge \frac{ds}{4} - \left( \delta d + 2c\sqrt{d\log k} + \frac{2d}{k}\right) n .
\end{eqnarray*}

For the last inequality we used the facts that $\delta \le \frac{1}{4}$ and $\ell < n$, and assumed that $c\sqrt{d\log k} \le \frac{d}{8}$.

Let us now require that $s \ge \frac{8}{d}(\delta d + 2c\sqrt{d\log k} + \frac{2d}{k})n = (8\delta + \frac{16c\sqrt{\log k}}{\sqrt{d}} + \frac{16}{k})n$, and then the number of edges with one endpoint in $S$ and no endpoint in $T$ is at least $\frac{ds}{8}$.
Consider  Step~2. In that step every edge with one endpoint in $S$ and no endpoint in $T$ needs to be dropped. The probability of this happening is at most $(1 - \frac{q}{d})^{ds/8} \simeq e^{-qs/8}$. Hence if $s = \frac{16n}{q}$, the probability is at most $e^{-2n}$ (for larger $s$ the probability is even smaller). This suffices for applying a union bound over all (at most) $3^n$ choices of $S$ and $T$, and conclude that with overwhelming probability, after Step~2, $G_1$ has a matching of size at least $(\frac{1}{2} - \frac{16}{q})n$.

Finally, some of the matched edges might be removed in Step~3, because one of their incident vertices had more than $q$ surviving edges. However, the expected number of surviving edges that a vertex has is at most $\frac{1}{2}\frac{q}{d}(1 + \delta)d \le \frac{5}{8}q$ (using our assumption that $\delta \le \frac{1}{4}$). Hence the probability that a vertex has more than $q$ surviving edges is at most $2^{-c' q}$ (for some universal constant $c' > 0$), and the expected number of matched edges that are lost due to such events is at most $2^{-c' q}n$. Therefore the expected size of the maximum matching in $G'$ is at least $(\frac{1}{2} - \frac{16}{q} - 2^{-c' q})n$

Let us now show how to choose the parameters for our proof. To obtain the desired lower bound on the size of the maximum matching, we require that $\frac{16}{q} + 2^{-c' q} \le \epsilon$. Choosing $q \ge \frac{c_3}{\epsilon}$ for $c_3$ slightly above $16$ suffices for this.
In addition, recall that we took $s = \frac{16n}{q}$ and hence $s \simeq \epsilon n$. In the course of the proof we required that:

\begin{enumerate}

\item $c\sqrt{d\log k} \le \frac{d}{8}$.

\item $\epsilon n \simeq s \ge (8\delta + \frac{16c\sqrt{\log k}}{\sqrt{d}} + \frac{16}{k})n$.

\end{enumerate}

For the second requirement, we may balance the three terms on the right hand side to each give a value of $\frac{\epsilon}{3}$. For the first term setting $c_2 = \frac{1}{24}$ suffices. For the third term we choose $k = \frac{48}{\epsilon}$. Then the first requirement translates to $d \ge 64c^2\log\frac{48}{\epsilon}$. For the second term we need a stronger requirement on $d$, namely,
$d \ge \frac{(48c)^2\log \frac{48}{\epsilon}}{\epsilon^2}$, and this governs the choice of $c_1$.
\end{proof}


\begin{corollary}\label{cor:fin}
There is a randomized reduction from the problem of approximating maximum matching in a graph $G$ with arbitrarily large average degree $d \ge q^2\log q$ and all degrees in the range $[(1 - \frac{1}{q})d, (1 + \frac{1}{q})d]$, to that of approximating maximum matching in a bipartite graph $G'$ of degree at most $q$, with an expected additive loss of at most $O(\frac{1}{q})$ in the approximation factor, and a multiplicative loss of $q$ in the number of strong probes.
\end{corollary}

\begin{proof}
The bound on the number of probes follows from Lemma~\ref{lem:simulation1}, and the bound on the approximation error follows from Theorem~\ref{thm:sparsify}
\end{proof}

\end{document}